\documentclass[smallextended]{svjour3} 
\usepackage{eurosym}
\usepackage{amsfonts, graphicx,amssymb,amsmath,color}
\usepackage[dvipdfm]{hyperref}
\hypersetup{citecolor=blue,colorlinks=true}
\usepackage[authoryear,comma]{natbib}

\begin{document}

\title{\textbf{Robust Tests for the Equality of Two Normal Means based on
the Density Power Divergence }}

\author{A. Basu \and A. Mandal \and N. Martin \and L. Pardo}

\institute{
A. Basu \and A. Mandal \at Indian Statistical Institute, Kolkata 700108, India
\and N. Martin \at Department of Statistics, Carlos III University of Madrid,
28903 Getafe (Madrid), Spain
\and L. Pardo \at Department of Statistics and O.R. Complutense University of
Madrid, 28040 Madrid, Spain} 

\titlerunning{Robust Tests for the Equality of Two Normal Means based on DPD} 
\authorrunning{Basu, A.; Mandal, A.; Martin, N. and Pardo, L.}
\date{September 28, 2014}
\maketitle

%
%
\begin{abstract}
Statistical techniques are used in all branches of science to determine the
feasibility of quantitative hypotheses. One of the most basic applications
of statistical techniques in comparative analysis is the test of equality of
two population means, generally performed under the assumption of normality.
In medical studies, for example, we often need to compare the effects of two
different drugs, treatments or preconditions on the resulting outcome. The
most commonly used test in this connection is the two sample $t$-test for
the equality of means, performed under the assumption of equality of
variances. It is a very useful tool, which is widely used by practitioners
of all disciplines and has many optimality properties under the model.
However, the test has one major drawback; it is highly sensitive to
deviations from the ideal conditions, and may perform miserably under model
misspecification and the presence of outliers. In this paper we present a
robust test for the two sample hypothesis based on the density power
divergence measure \citep{MR1665873}, and show that it can be a great
alternative to the ordinary two sample $t$-test. The asymptotic properties
of the proposed tests are rigorously established in the paper, and their
performances are explored through simulations and real data analysis.
\end{abstract}
%
%
\bigskip\bigskip


\noindent\underline{\textbf{keywords and phrases}}: Robustness, Density Power
Divergence, Hypothesis Testing.
%
%
\section{Introduction: Motivation and Background}
%
%
In many scientific studies, often the main problem of interest is to compare
different population groups. In medical studies, for example, the primary
research problem could be to test for the difference between the location
parameters of two different populations receiving two different drugs, treatments or
therapy, or having two different preconditions. The normal
distribution often provides the basic setup for statistical analyses in 
medical studies (as well as in other disciplines). Inference
procedures based on the sample mean, the standard deviation and the one and
two-sample $t$-tests are often the default techniques for the scenarios where
they are applicable. In particular, the two sample $t$-test is the most
popular technique in testing for the equality of two means, performed under
the assumption of equality of variances. Its applicability in real life
situations is, however, tempered by the known lack of robustness of this
test against model perturbations. Even a small deviation from the ideal
conditions can make the test completely meaningless and lead to nonsensical
results. This problem is caused by the fact that the $t$-test is based on
the classical estimates of the location and scale parameters (the sample
mean and the sample standard deviation). Large outliers tend to distort the 
mean and inflate the standard deviation. This may lead to false results of 
both types, i.e. detecting a difference when there isn't one, and failing to
detect a true significance. 
 
In this paper we are going to develop a class of robust tests for the two
sample problem which evolves from an appropriate minimum distance technique
in a natural way. This class of tests is indexed by two real parameters $%
\beta $ and $\gamma $, and we will constrain each of these parameters to lie
within the $[0,1]$ interval. Our general minimum distance approach will allow
us to study the likelihood ratio test in an asymptotic sense, as the likelihood 
ratio test is asymptotically equivalent to the test generated by the parameters $%
\beta =\gamma =0$. Normally we will
work with the one parameter family of test statistics corresponding to $%
\beta =\gamma $; the outlier stability of the proposed tests increase with
the tuning parameter $\gamma $.

Let $X$ and $Y$ be independent random variables whose distributions are
modeled as normals having unknown means $\mu_1$ and $\mu_2$,
respectively, with an unknown but common variance $\sigma^2$. We are
interested in testing the null hypothesis
%
%
\begin{equation}
H_{0}:\mu_1=\mu_2\text{ against }H_{1}:\mu_1\neq \mu_2,
\label{EQ:0}
\end{equation}
%
%
under the above set up. It is well known that the exact two sample $t$-test (which is equivalent 
to the likelihood ratio test) rejects the null hypothesis in (\ref{EQ:0}) if and only if
%
%
\begin{equation*}
t=\frac{\left\vert \bar{X}-\bar{Y}\right\vert }{S_{p}\sqrt{\frac{1%
}{n_1}+\frac{1}{n_2}}}>t_{\frac{\alpha }{2}}(n_1+n_2-2),
\end{equation*}
%
%
where $\bar{X}$ and $\bar{Y}$ are the sample means corresponding
to the random samples $X_{1},X_{2},\ldots ,X_{n_1}$ and $%
Y_{1},Y_{2},\ldots ,Y_{n_2}$ obtained from the two distributions,
%
%
\begin{equation*}
S_{p}^{2}=\frac{(n_1-1)S_{1}^{2}+(n_2-1)S_{2}^{2}}{n_1+n_2-2},
\end{equation*}
%
%
%
%
\begin{equation*}
S_{1}^{2}=\frac{1}{n_1-1}\sum_{i=1}^{n_1}\left( X_i-\bar{X}%
\right) ^{2},\quad S_{2}^{2}=\frac{1}{n_2-1}\sum_{i=1}^{n_2}\left( Y_{i}-%
\bar{Y}\right) ^{2},
\end{equation*}%
%
%
and $t_{\frac{\alpha }{2}}(n_1+n_2-2)$ is the $100(1-\frac{\alpha }{2})$-th
quantile of the $t$-distribution with $n_1+n_2-2$ degrees of freedom. The $t$-test 
is the uniformly most powerful unbiased and
invariant test for this hypothesis. Testing the equality of means of
independent normal populations with unknown variances which are not
necessarily equal, is referred to as the Behrens-Fisher problem.

In this paper we will use the density power divergence (DPD) measure \citep{MR1665873},
which provides a natural robustness option for many standard
inference problems. The density power divergence and its variants have been
successfully used by many authors in a variety of inference problems; see,
eg. \cite{MR1859416}, \cite{MR2299175,MR2466551}, \cite{MR3011625,basu2013}, \cite{MR3117102}.
However, the two sample problem requires a
non-trivial extension of the currently existing techniques. Our purpose in
this paper is to derive the asymptotic properties of the class of two sample
tests based on the density power divergence and demonstrate their robust
behavior in practical situations.

\bigskip

\noindent \textbf{Example 1 (Cloth Manufacturing data)}: In order to
emphasize the need for applications early, we now present a motivational
example. This example illustrates the use of quality control methods practiced in a
clothing manufacturing plant. Levi-Strauss manufactures clothing from cloth
supplied by several mills. The data used in this example (see Table \ref%
{TAB:Staudte_Sheather}) are for two of these mills and were obtained from
the quality control department of the Levi plant in Albuquerque, New Mexico
(\citealp{lambert1987introduction}, p. 86). In order to maintain the anonymity of these two
mills we have coded them $A$ and $B$. A measure of wastage due to defects in
cloth and so on is called \emph{run-up}. It is quoted as percentage of
wastage per week and is measured relative to computerized layouts of
patterns on the cloth. Since the people working in the plant can often beat
the computer in reducing wastage by laying out the patterns by hand, it is
possible for run-up to be negative. From the viewpoint of quality control,
it is desirable not only that the run-up be small but that the quality from
week to week be fairly consistent. There are 22 measurements on run-up for
each of the two mills and they are presented in Table \ref{TAB:Staudte_Sheather}. 
The $t$-test for the equality of the two means against the two-sided
alternative has a $p$-value of 0.3428 and fails to reject the null
hypothesis; however, when the presumed outliers (presented in bold fonts in
Table \ref{TAB:Staudte_Sheather}) are removed from the dataset, the same
two-sample $t$-test produces a $p$-value of 0.0308, leading to clear
rejection. Choosing $\beta = \gamma$ to be the only parameter, the $p$%
-values of the DPD tests (to be developed in the next section) for testing
the same hypotheses are presented in Figure \ref%
{fig:Staudte_Sheather_book_p_val} as a function of $\gamma$. It is observed
that the $p$-values of the tests with the full data and those with the
outlier deleted data are practically identical for $\gamma = 0.2$ or larger,
and lead to solid rejection. Thus, while the outliers mask the significance
in case of the two sample $t$-test, the more robust DPD tests are able to
capture the same.

%
%
\begin{table}[tbp]
\caption{Cloth Manufacturing data.}
\label{TAB:Staudte_Sheather}
\begin{center}
\begin{tabular}{lrrrrrrrrrrr}
\hline
Mill A & $0.12$ & $1.01$ & $-0.20$ & $0.15$ & $-0.30$ & $-0.07$ & $0.32$ & $%
0.27$ & $-0.32$ & $-0.17$ & $0.24$ \\
& $0.03$ & $0.35$ & $-0.08$ & $\bf{2.94}$ & $0.28$ & $1.30$ & $\bf{4.27}$ & $0.14$ & $%
0.30$ & $0.24$ & $0.13$ \\ \hline
Mill B & $1.64$ & $-0.60$ & $-1.16$ & $-0.13$ & $0.40$ & $1.70$ & $0.38$ & $%
0.43$ & $1.04$ & $0.42$ & $0.85$ \\
& $0.63$ & $0.90$ & $0.71$ & $0.43$ & $1.97$ & $0.30$ & $0.76$ & $\bf{7.02}$ & $%
0.85$ & $0.60$ & $0.29$ \\ \hline
\end{tabular}%
\end{center}
\end{table}
%
%
%
%
\begin{figure}
\centering
{\includegraphics[height=6.5cm, width=14cm]{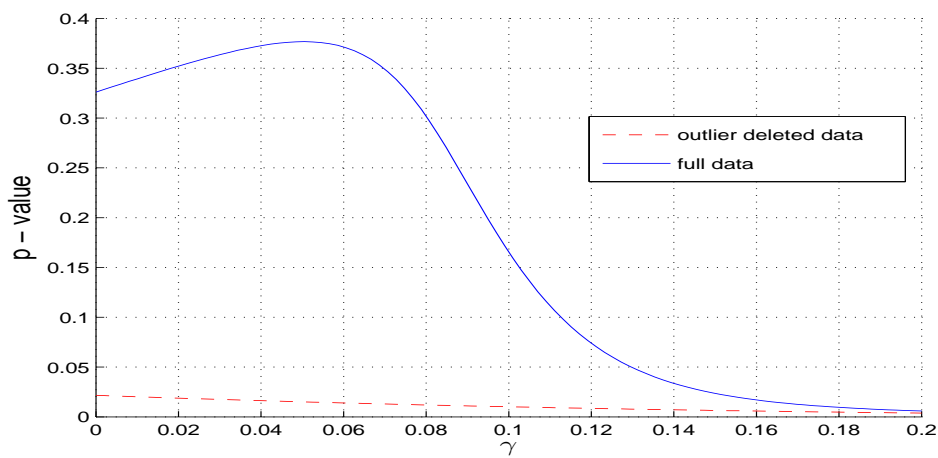}}
\caption{The $p$-values of the DPD tests for the Cloth Manufacturing data for different values of $\gamma$. The solid line represents the full data analysis, while the dashed line represents the outlier deleted case.}
\label{fig:Staudte_Sheather_book_p_val}
\end{figure}
%
%


Our primary motivation for studying the alternatives of the two sample $t$-test has been the 
need for developing such a test in the context of examples relating to medical data. However, 
examples abound in practically all scientific disciplines showing that this is a real necessity
which is certainly not restricted to the medical field. The example considered above is one such, 
where the context does not have anything directly to do with a medical problem, but the importance 
of the problem and the need for a robust solution can immediately be appreciated. 

The rest of the paper is organized as follows: In Section \ref{SEC:MDPDE} the asymptotic distribution
of the minimum DPD estimators in the two sample situation is described. In Section \ref{SEC:Test}
we introduce our robust two sample test statistic and develop the necessary theory. A large number 
of real data examples and extensive simulation results are presented in Section \ref{SEC:numerical}. 
Finally Section \ref{SEC:concluding} has some concluding remarks. 

\section{The Minimum DPD Estimator: Asymptotic Distribution}\label{SEC:MDPDE}

For any two probability density functions $f$ and $g$, the density power
divergence measure is defined, as the function of a single tuning
parameter $\beta \geq 0$, as
\begin{equation}
d_{\beta}(g,f)=\left\{
\begin{array}
[c]{ll}%
\int\left\{  f^{1+\beta}(x)-\left(  1+\frac{1}{\beta}\right)  f^{\beta
}(x)g(x)+\frac{1}{\beta}g^{1+\beta}(x)\right\}  dx, & \text{for}%
\mathrm{~}\beta>0,\\[2ex]%
\int g(x)\log\left(  \displaystyle\frac{g(x)}{f(x)}\right)  dx, &
\text{for}\mathrm{~}\beta=0.
\end{array}
\right.  \label{EQ:definition_DPD}%
\end{equation}
%
%
Let $X_{1},X_{2},\ldots ,X_n$ be a random sample of size $n$
from a $\mathcal{N}(\mu,\sigma^2)$ distribution, where both
parameters are unknown. Let $f_{\mu,\sigma }(x)$ represent the density function of a 
$\mathcal{N}(\mu,\sigma^2)$ variable. For a given $\beta $,
we get the minimum density power divergence estimators (MDPDEs) $\widehat{\mu }%
_{\beta }$ and $\widehat{\sigma }_{\beta}$ of $\mu$ and $\sigma$ by
minimizing the following function over $\mu$ and $\sigma$ 
%
%
\begin{equation}
\int_{\mathbb{R}}f_{\mu,\sigma }^{1+\beta }(x)dx-\left( 1+\frac{1}{%
\beta }\right) \frac{1}{n}\sum_{i=1}^{n}f_{\mu,\sigma }^{\beta
}(X_i),\text{\qquad for }\beta >0,  \label{1}
\end{equation}
%
%
and
%
%
\begin{equation}
-\frac{1}{n}\sum_{i=1}^{n}\log f_{\mu,\sigma }(X_i),\text{%
\qquad for }\beta =0.  \label{EQ:1.0}
\end{equation}
%
%
For $\beta =0$, the objective function in (\ref{EQ:1.0}) is the
negative of the usual log likelihood and has the classical maximum
likelihood estimator as the minimizer. For a
normal density the function in (\ref{1}) simplifies to
%
%
\begin{equation*}
h_{n,\beta }(\mu,\sigma )=\frac{1}{\sigma ^{\beta }(2\pi )^{\frac{%
\beta }{2}}}\left\{ \frac{1}{\left( 1+\beta \right) ^{3/2}}-\frac{1}{%
n\beta }\sum_{i=1}^{n}\exp \left( -\frac{1}{2}\left( \frac{X_i-\mu}{{\sigma }}\right) ^{2}\beta \right) \right\} .
\end{equation*}
%
%
In order to get $\widehat{\mu }_{\beta }$ and $\widehat{\sigma}_\beta $, we have to solve the estimating
equation 
\begin{equation}
\mathbf{h'}_{n,\beta }(\widehat{\mu }_{\beta },\widehat{\sigma}_\beta ) =%
\begin{pmatrix}
_{1}h'_{n,\beta }(\widehat{\mu }_{\beta },\widehat{\sigma}_\beta ) \\
_{2}h'_{n,\beta }(\widehat{\mu }_{\beta },\widehat{\sigma}_\beta )%
\end{pmatrix} = \boldsymbol{0}_2, 
\label{h1}
\end{equation}
where
\begin{equation}
 _{1}h_{n,\beta }^{\prime }(\widehat{\mu }_{\beta },\widehat{\sigma}_\beta ) =\left. \frac{%
\partial h_{n,\beta }(\mu,\widehat{\sigma}_\beta )}{\partial \mu}%
\right\vert _{\mu=\widehat{\mu }_{\beta }},\qquad _{2}h_{n,\beta }^{\prime }(\widehat{\mu }_{\beta },\widehat{\sigma}_\beta  )
=\left. \frac{\partial h_{n,\beta }(\widehat{\mu }_{\beta },\sigma )}{%
\partial \sigma }\right\vert _{\sigma =\widehat{\sigma}_\beta },
\end{equation}
and $\boldsymbol{0}_2$ represents a zero vector of length 2. 
We denote
\begin{equation*}
\mathbf{H}_{n,\beta }( \mu_0,\sigma_0 ) =\left(
\begin{array}{cc}
_{11}h_{n,\beta }^{\prime \prime }\left( \mu_0,\sigma_0\right)  &
_{12}h_{n,\beta }^{\prime \prime }\left( \mu_0,\sigma_0\right)
\\
_{21}h_{n,\beta }^{\prime \prime }\left( \mu_0,\sigma_0\right)  &
_{22}h_{n,\beta }^{\prime \prime }\left( \mu_0,\sigma_0\right)
\end{array}%
\right),
\end{equation*}%
%
%
where
%
%
\begin{align*}
_{11}h_{n,\beta }^{\prime \prime }\left( \mu_0,\sigma_0\right) &
=\left. \dfrac{\partial ^{2}h_{n,\beta }\left( \mu,\sigma
_{0}\right) }{\partial \mu^{2}}\right\vert _{\mu =\mu_0},\qquad
_{12}h_{n,\beta }^{\prime \prime }\left( \mu_0,\sigma_0\right)
=\left. \dfrac{\partial ^{2}h_{n,\beta }\left( \mu,\sigma \right) }{%
\partial \mu\partial \sigma }\right\vert _{\mu =\mu_0,\sigma
=\sigma_0}, \\
_{21}h_{n,\beta }^{\prime \prime }\left( \mu_0,\sigma_0\right) &
=\left. \dfrac{\partial ^{2}h_{n,\beta }\left( \mu,\sigma \right) }{%
\partial \sigma \partial \mu}\right\vert _{\mu =\mu_0,\sigma
=\sigma_0},\qquad _{22}h_{n,\beta }^{\prime \prime }\left( \mu
_{0},\sigma_0\right) =\left. \dfrac{\partial ^{2}h_{n,\beta }\left(
\mu_0,\sigma \right) }{\partial \sigma^2}\right\vert _{\sigma =\sigma
_{0}}.
\end{align*}%
Using a Taylor series expansion of the function in equation (\ref{h1}), it is easy to show that
\begin{eqnarray}
 \sqrt{n}
 \begin{pmatrix}
  \widehat{\mu }_{\beta } - \mu_0\\
  \widehat{\sigma}_\beta - \sigma_0
 \end{pmatrix}
&=& \sqrt{n} \mathbf{H}_{n,\beta }^{-1}( \mu_0,\sigma_0 ) \boldsymbol{h}'_{n,\beta }(\mu_0,\sigma_0 ) + o_p(1) \nonumber\\
&=& \sqrt{n} \mathbf{J}_\beta^{-1}( \sigma_0 ) \boldsymbol{h}'_{n,\beta }(\mu_0,\sigma_0 ) + o_p(1),
\label{muSigma}
\end{eqnarray}
where
\begin{equation}
\boldsymbol{J}_{\beta }(\sigma_0) = \lim_{n \rightarrow \infty }%
\mathbf{H}_{n,\beta } ( \mu_0,\sigma_0 ) 
= \frac{1}{%
\sqrt{1+\beta }\left( 2\pi \right) ^{\beta /2}\sigma_0 ^{2+\beta }}\left(
\begin{array}{cc}
\frac{1}{1+\beta } & 0 \\
0 & \frac{\beta ^{2}+2}{\left( 1+\beta \right) ^{2}}%
\end{array}%
\right).
\end{equation}
The joint distribution of $\widehat{\mu }_{\beta }$ and $\widehat{\sigma}_\beta$ then follows (see \citealp{MR3011625}) from the result that
\begin{equation}
\sqrt{n}\mathbf{h'}_{n,\beta}(\mu _0,\sigma_0)\underset{%
n\rightarrow \infty }{\overset{\mathcal{L}}{\longrightarrow }}\mathcal{N}%
\left( \boldsymbol{0}_{2},\boldsymbol{K}_{\beta }(\sigma_0)\right) ,
\label{1.1}
\end{equation}
where
\begin{eqnarray}
\boldsymbol{K}_{\beta }(\sigma_0) &=& \left( K_{ij,\beta
}(\sigma_0)\right) _{i,j=1,2}  \nonumber\\
&=&
\frac{1}{\sigma_0 ^{2+2\beta
}\left( 2\pi \right) ^{\beta }}\left( \frac{1}{(1+2\beta )^{3/2}}\left(
\begin{array}{cc}
1 & 0 \\
0 & \frac{4\beta ^{2}+2}{1+2\beta }%
\end{array}%
\right) -\left(
\begin{array}{cc}
0 & 0 \\
0 & \frac{\beta ^{2}}{(1+\beta )^{3}}%
\end{array}%
\right) \right) .
\label{2}
\end{eqnarray}
We will use the above results to obtain the MDPDEs of the parameters in the two sample setup mentioned below.

Suppose $X_{1},X_{2},\ldots ,X_{n_1}$ is a random sample of size $n_1$
from $X$ which has a $\mathcal{N}(\mu_1,\sigma^2)$ distribution, and $%
Y_{1},Y_{2},\ldots ,Y_{n_2}$ is a random sample of size $n_2$ from $Y$
which has a $\mathcal{N}(\mu_2,\sigma^2)$ distribution; all three
parameters are unknown. Let $f_{\mu_1,\sigma }(x)$ and $f_{\mu
_{2},\sigma }(y)$ be the density functions of $X$ and $Y$ respectively.
Let us denote the set of unknown parameters by $\boldsymbol{\eta }=(\mu_1,\mu
_{2},\sigma )^{T}$. The MDPDE of $\boldsymbol{\eta }$, denoted by 
$\widehat{\boldsymbol{\eta }}_{\beta }=(\widehat{\mu}_{1\beta },\widehat{\mu}_{2\beta
},\widehat{\sigma}_{\beta })^{T}$, is obtained by minimizing the following function
%
%
\begin{equation}
h_{n_1,n_2,\beta }(\boldsymbol{\eta })=\frac{1}{n_1+n_2}%
\left( n_{1\text{ }}h_{n_1,\beta }(\mu_1,\sigma )+n_{2\text{ }%
}h_{n_2,\beta }\left( \mu_2,\sigma \right) \right) .
\label{hn12}
\end{equation}
%
%
It may be noticed that $\hat{\mu}_{1\beta }$ is based only on the first term of the above function,
and similarly $\hat{\mu}_{2\beta }$ depends only on the second term. Therefore, the estimating 
equations are given by $_{1}h_{n_{i},\beta }^{\prime }\left( \mu _{i},\sigma \right) =0$,
$i=1,2$, and $_{2}h_{n_1,n_2,\beta }^{\prime }(\boldsymbol{\eta } )=0$, 
where
\begin{equation}
_{2}h_{n_1,n_2,\beta }^{\prime }(\boldsymbol{\eta })=\frac{%
\partial h_{n_1,n_2,\beta }(\boldsymbol{\eta } )}{\partial \sigma }%
=\frac{1}{n_1+n_2}\left( n_1\,\allowbreak _{2}h_{n_1,\beta }^{\prime
}\left( \mu_1,\sigma \right) +n_2\,\allowbreak _{2}h_{n_2,\beta
}^{\prime }\left( \mu_2,\sigma \right) \right) .
\label{h2n12}
\end{equation}%
%
%

For $\beta =0$, the above equations can be explicitly solved to get the
MDPDEs for this case. It is easily seen that $\widehat{\mu }_{10}=\bar{X%
}$ and $\widehat{\mu }_{20}=\bar{Y}$. Moreover, using equation (\ref{EQ:1.0})
we get from (\ref{hn12}) 
%
%
\begin{align*}
& h_{n_1,n_2,\beta =0}(\widehat{\boldsymbol{\eta }}_0 ) \\
& =-\frac{1}{n_1+n_2}\left( n_1\frac{1}{n_1}\log
\prod_{i=1}^{n_1}f_{\widehat{\mu }_{10},\widehat{\sigma}_0 }(X_i)+n_2\frac{1}{n_2}\log
\prod_{i=1}^{n_2}f_{\widehat{\mu }_{20},\widehat{\sigma}_0 }(Y_{i})\right) \\
& =\frac{1}{n_1+n_2}\left( (n_1+n_2)\log \widehat{\sigma}_0 +\sum_{i=1}^{n_1}%
\frac{\left( X_i-\bar{X}\right) ^{2}}{2\widehat{\sigma}_0^2}+\sum_{i=1}^{n_2}%
\frac{\left( Y_{i}-\bar{Y}\right) ^{2}}{2\widehat{\sigma}_0^2}+(n_1+n_2)\log
\sqrt{2\pi }\right) .
\end{align*}
%
%
So,
%
%
\begin{equation*}
_{2}h_{n_1,n_2,\beta }^{\prime }(\widehat{\boldsymbol{\eta }}_0 ) 
=\frac{1}{\widehat{\sigma}_0 }-\frac{1}{%
\widehat{\sigma}_0 ^{3}(n_1+n_2)}\left\{ (n_1-1)S_{1}^{2} + (n_2-1)S_{2}^{2}\right\}  ,
\end{equation*}%
%
%
which leads to the solution
%
%
\begin{equation}
\widehat{\sigma }_{0}=\left( \frac{(n_1-1)S_{1}^{2}+(n_2-1)S_{2}^{2}}{%
n_1+n_2}\right) ^{\frac{1}{2}}.  \label{sig0}
\end{equation}
Therefore, for $\beta=0$ the MDPDEs turn out to be the MLEs of the corresponding parameters. The 
following theorem gives the asymptotic distribution of the MDPDE of $\boldsymbol{\eta }$ for a 
given $\beta$.
\begin{theorem}
\label{Th0}We consider two normal populations with unknown means $\mu_1$
and $\mu_2$ and unknown but common variance $\sigma^2.$ Let
\begin{equation}
w=\lim_{n_1,n_2\rightarrow \infty }\frac{n_1}{n_1+n_2}
\label{EQ:w}
\end{equation}
%
%
be the limiting proportion of observations from the first population in
the whole sample. We assume that $w \in (0,1)$. Then, the minimum density power divergence estimator $\widehat{\boldsymbol{\eta }}_{\beta }$ of $\boldsymbol{\eta}$ 
has the asymptotic distribution given by
%
%
\begin{equation}
\sqrt{\frac{n_1n_2}{n_1+n_2}}(\widehat{\boldsymbol{\eta }}_{\beta }-%
\boldsymbol{\eta }_{0})\underset{n_1,n_2\rightarrow \infty }{\overset{%
\mathcal{L}}{\longrightarrow }}\mathcal{N}\left( \boldsymbol{0}_{3},%
\boldsymbol{\Sigma }_{w,\beta }(\sigma_0)\right) ,  \label{eqTh0}
\end{equation}%
%
%
where $\boldsymbol{\eta }_{0}=(\mu _{10},\mu _{20},\sigma_0)^{T}$ is the
true value of $\boldsymbol{\eta }$, and
%
%
\begin{equation}
\boldsymbol{\Sigma }_{w,\beta }(\sigma_0)=\sigma_0^{2}\left(
\begin{array}{ccc}
\left( 1-w\right) \frac{\left( \beta +1\right) ^{3}}{\left( 2\beta +1\right)
^{\frac{3}{2}}} & 0 & 0 \\
0 & w\frac{\left( \beta +1\right) ^{3}}{\left( 2\beta +1\right) ^{\frac{3}{2}%
}} & 0 \\
0 & 0 & w\left( 1-w\right) \frac{\left( \beta +1\right) ^{5}}{\left( \beta
^{2}+2\right) ^{2}}\left( \frac{4\beta ^{2}+2}{(1+2\beta )^{5/2}}-\frac{%
\beta ^{2}}{(1+\beta )^{3}}\right)%
\end{array}%
\right) .  \label{8}
\end{equation}%

\end{theorem}

\begin{proof}
See Appendix.
\hspace*{\fill}\bigskip
\end{proof}

\section{The Asymptotic Distribution of the DPD Test Statistic}\label{SEC:Test}

Let $f_{\mu_1,\sigma _{1}}(x)$ and $f_{\mu_2,\sigma _{2}}(y)$ be the
density functions of $X\sim \mathcal{N}(\mu_1,\sigma _{1})$ and $Y\sim
\mathcal{N}(\mu_2,\sigma _{2})$ respectively. The density power divergence
measure between the densities of $X$ and $Y$, for $\gamma >0$, is given by
%
%
\begin{align*}
d_{\gamma }(f_{\mu_1,\sigma _{1}},f_{\mu_2,\sigma _{2}}) =&\frac{1}{%
\sigma _{2}^{\gamma }\sqrt{1+\gamma }\left( 2\pi \right) ^{\gamma /2}}+%
\frac{1}{\gamma\sigma _{1}^{\gamma }\sqrt{1+\gamma }\left( 2\pi \right)
^{\gamma /2}} \\
&  -\frac{\gamma +1}{\gamma \sigma _{2}^{\gamma -1}(\gamma \sigma
_{1}^{2}+\sigma _{2}^{2})^{1/2}\left( 2\pi \right) ^{\gamma/2 }} \\
&  \times \exp \left\{ \frac{1}{2}\left[ -\left( \tfrac{\mu_2^{2}}{\left(
\frac{\sigma _{2}}{\sqrt{\gamma }}\right) ^{2}}+\tfrac{\mu_1^{2}}{\sigma
_{1}^{2}}\right) +\tfrac{\left( \sigma _{1}^{2}\mu_2+\mu_1\left( \frac{%
\sigma _{2}}{\sqrt{\gamma }}\right) ^{2}\right) ^{2}}{\left( \sigma
_{1}^{2}+\left( \frac{\sigma _{2}}{\sqrt{\gamma }}\right) ^{2}\right) \left(
\frac{\sigma _{2}}{\sqrt{\gamma }}\right) ^{2}\sigma _{1}^{2}}\right]
\right\} ,
\end{align*}
%
%
and for $\gamma =0$%
%
%
\begin{equation*}
d_\gamma(f_{\mu_1,\sigma _{1}},f_{\mu_2,\sigma _{2}})=\log {\frac{\sigma
_{2}}{\sigma _{1}}}-\frac{1}{2}+\frac{\sigma _{1}^{2}}{2\sigma _{2}^{2}}+%
\frac{1}{2\sigma _{2}^{2}}(\mu_1-\mu_2)^{2}.
\end{equation*}%
%
%
To test the null hypothesis given in (\ref{EQ:0}), under the
assumption that $\sigma _{1}=\sigma _{2}=\sigma $, we will consider the
divergence between the two normal populations with the estimated parameters;
this yields
\begin{equation}
d_{\gamma }(f_{\widehat{\mu }_{1\beta },\widehat{\sigma}_\beta },f_{%
\widehat{\mu }_{2\beta },\widehat{\sigma}_\beta })=\left\{
\begin{array}{ll}
\frac{\sqrt{1+\gamma }}{\gamma \left( \sqrt{2\pi }\widehat{\sigma }_{\beta
}\right) ^{\gamma }}\left[ 1-\exp \left\{ -\frac{\gamma }{2(\gamma
+1)}\left( \frac{\widehat{\mu }_{1\beta }-\widehat{\mu }_{2\beta }}{\widehat{%
\sigma }_{\beta }}\right) ^{2}\right\} \right] , & \text{ for }\gamma >0, \\
\frac{1}{2}\left( \frac{\widehat{\mu }_{1\beta }-\widehat{\mu }_{2\beta }}{%
\widehat{\sigma}_\beta }\right) ^{2}, &  \text{ for }\gamma =0.%
\end{array}%
\right.  \label{EQ:initial_statistic}
\end{equation}%
%
%
Naturally, we will reject the null hypothesis for large values of $d_{\gamma }(f_{\widehat{\mu }%
_{1\beta },\widehat{\sigma}_\beta },f_{\widehat{\mu }_{2\beta },\widehat{%
\sigma }_{\beta }})$. 
To propose the test in a very general setup we have considered two possibly distinct tuning parameters $\gamma$
and $\beta$ in the 
above expression; the parameter $\gamma$ represents the tuning parameters of the divergence, and the parameter 
$\beta$ represents the tuning parameter of the MDPDEs. 
In order to determine the critical region of this
test we will find (later in Theorem \ref{Th2}) the asymptotic null distribution of the test statistic based on (\ref%
{EQ:initial_statistic}), standardized with a suitable scaling constant involving $n_1$ and $n_2$.

\begin{theorem}
\label{Th1} For $\gamma >0$, let us define
$\boldsymbol{t}_{\gamma }\left( \boldsymbol{\eta }\right) =(t_{\gamma, 1}(%
\boldsymbol{\eta }),t_{\gamma, 2}(\boldsymbol{\eta }),t_{\gamma, 3}(%
\boldsymbol{\eta }))^{T}$, with%
%
%
\begin{align}
t_{\gamma, 1}(\boldsymbol{\eta })& =\frac{\frac{\mu_1-\mu_2}{\sigma }}{%
\sqrt{1+\gamma }\left( \sqrt{2\pi }\right) ^{\gamma }\sigma ^{\gamma +1}}%
\exp \left\{ -\frac{1}{2}\tfrac{\gamma }{\gamma +1}\left( \tfrac{\mu
_{1}-\mu_2}{\sigma }\right) ^{2}\right\} ,  \label{t1} \\
t_{\gamma, 2}(\boldsymbol{\eta })& =-t_{1}(\boldsymbol{\eta }),  \label{t2} \\
t_{\gamma, 3}(\boldsymbol{\eta })& = - \tfrac{\sqrt{1+\gamma }}{\left( \sqrt{%
2\pi }\right) ^{\gamma }\sigma ^{\gamma +1}}\left[ 1-\left( 1 - \tfrac{1}{%
1+\gamma }\left( \tfrac{\mu_1-\mu_2}{\sigma }\right) ^{2}\right) \exp
\left\{ -\tfrac{1}{2}\tfrac{\gamma }{\gamma +1}\left( \tfrac{\mu_1-\mu
_{2}}{\sigma }\right) ^{2}\right\} \right] .  \label{t3}
\end{align}%
%
%
Then, for $w \in (0,1)$  as defined in (\ref{EQ:w}) we have
\begin{equation}
\sqrt{\frac{n_1n_2}{n_1+n_2}}\left( d_{\gamma }(f_{\widehat{\mu }%
_{1\beta },\widehat{\sigma}_\beta },f_{\widehat{\mu }_{2\beta },\widehat{%
\sigma }_{\beta }})-d_{\gamma }(f_{\mu _{10},\sigma_0},f_{\mu
_{20},\sigma_0})\right) \underset{n_1,n_2\rightarrow \infty }{\overset%
{\mathcal{L}}{\longrightarrow }}\mathcal{N}\left( 0, \sigma_\gamma^2
\right) ,  \label{resTh1}
\end{equation}
%
%
where 
%
%
\begin{equation}
 \sigma_\gamma^2 = \boldsymbol{t}_{\gamma
}^{T}\left( \boldsymbol{\eta }_{0}\right) \boldsymbol{\Sigma }_{w,\beta
}(\sigma_0)\boldsymbol{t}_{\gamma }\left( \boldsymbol{\eta }_{0}\right),
\label{sigma_gamma}
\end{equation}
%
%
and $\boldsymbol{\Sigma }_{w,\beta }(\sigma_0)$ is given in (\ref{8}).
\end{theorem}

\begin{proof}
See Appendix.
\hspace*{\fill}\bigskip
\end{proof}

Notice that $\boldsymbol{t}_{\gamma }^{T}\left( \boldsymbol{\eta }%
_{0}\right) \boldsymbol{\Sigma }_{w,\beta }(\sigma_0)\boldsymbol{t}%
_{\gamma }\left( \boldsymbol{\eta }_{0}\right) \geq 0$. If $\mu _{10}\neq
\mu _{20}$, we  observe that $\boldsymbol{t}_{\gamma }\left( \boldsymbol{%
\eta }_{0}\right) \neq \boldsymbol{0}_{3}$, and since $\boldsymbol{\Sigma }%
_{w,\beta }(\sigma_0)$ is positive definite matrix, we have $\boldsymbol{t}%
_{\gamma }^{T}\left( \boldsymbol{\eta }_{0}\right) \boldsymbol{\Sigma }%
_{w,\beta }(\sigma_0)\boldsymbol{t}_{\gamma }\left( \boldsymbol{\eta }%
_{0}\right) >0$. But for $\mu _{10}=\mu _{20}$,  $\boldsymbol{t}%
_{\gamma }\left( \boldsymbol{\eta }_{0}\right) =\boldsymbol{0}_{3}$, and
hence $\boldsymbol{t}_{\gamma }^{T}\left( \boldsymbol{\eta }_{0}\right)
\boldsymbol{\Sigma }_{w,\beta }(\sigma_0)\boldsymbol{t}_{\gamma }\left(
\boldsymbol{\eta }_{0}\right) =0$. Therefore, to get the asymptotic distribution 
of the test statistic under the null hypothesis we need a higher order scaling involving $n_1$ and $n_2$
to the quantity given in (\ref{EQ:initial_statistic}).

\begin{theorem}
\label{Th2}Let $w\in (0,1)$ as defined in (\ref{EQ:w}) and $\gamma >0$. Then, under the null hypothesis, we have
%
%
\begin{equation}
S_{\gamma }\left( \widehat{\mu }_{1\beta },\widehat{\mu }_{2\beta },\widehat{%
\sigma }_{\beta }\right) =\frac{2n_1n_2}{n_1+n_2}\frac{d_{\gamma
}(f_{\widehat{\mu }_{1\beta },\widehat{\sigma}_\beta },f_{\widehat{\mu }%
_{2\beta },\widehat{\sigma}_\beta })}{\lambda _{\beta ,\gamma
}\,\allowbreak (\widehat{\sigma}_\beta )}\underset{n_1,n_2\rightarrow
\infty }{\overset{\mathcal{L}}{\longrightarrow }}\chi ^{2}(1),
\label{EQ:test}
\end{equation}%
%
%
where
%
%
\begin{equation}
\lambda _{\beta ,\gamma }(\widehat{\sigma}_\beta )=\frac{\left( \beta
+1\right) ^{3}\left( 2\beta +1\right) ^{-\frac{3}{2}}}{\widehat{\sigma }%
_{\beta }^{\gamma }\left( 2\pi \right) ^{\frac{\gamma }{2}}\left( \gamma
+1\right) ^{\frac{1}{2}}}.
\label{lambda1}
\end{equation}
%
%
\end{theorem}

\begin{proof}
See Appendix.
\hspace*{\fill}\bigskip
\end{proof}

\noindent
The above result indicates that the density power divergence test for the
hypothesis in (\ref{EQ:0}) can be based on the statistic $S_{\gamma }\left(
\widehat{\mu }_{1\beta },\widehat{\mu }_{2\beta },\widehat{\sigma }_{\beta
}\right) $, where the critical region corresponding to significance level $%
\alpha $ is given by the set of points satisfying
%
%
\begin{equation*}
S_{\gamma }\left( \widehat{\mu }_{1\beta },\widehat{\mu }_{2\beta },\widehat{%
\sigma }_{\beta }\right) >\chi _{\alpha }^{2}(1).
\end{equation*}%
%
%
%

Using the result of Theorem \ref{Th1} we can get an approximation of the power function of the test statistic. We consider $\mu _{10}\neq \mu _{20}$. 
In the following we will let $\lambda$ denote the quantity defined in equation (\ref{lambda1}) to keep the notation simple. The power function is then given by
%
%
\begin{eqnarray*}
 \eta_{\gamma,\beta}(\mu _{10}, \mu _{20}, \sigma_0) &=& P\left( S_{\gamma }\left( \widehat{\mu }_{1\beta },\widehat{\mu }_{2\beta },\widehat{%
\sigma }_{\beta }\right) >\chi _{\alpha }^{2}(1) \right) \\
&=& P\left( \frac{2}{\lambda}\frac{n_1n_2}{n_1+n_2} d_{\gamma }(f_{\widehat{\mu }%
_{1\beta },\widehat{\sigma}_\beta },f_{\widehat{\mu }_{2\beta },\widehat{%
\sigma }_{\beta }}) >\chi _{\alpha }^{2}(1) \right) \\
&=& P\Bigg( \sqrt{\frac{n_1n_2}{n_1+n_2}}\left( d_{\gamma }(f_{\widehat{\mu }%
_{1\beta },\widehat{\sigma}_\beta },f_{\widehat{\mu }_{2\beta },\widehat{%
\sigma }_{\beta }})-d_{\gamma }(f_{\mu _{10},\sigma_0},f_{\mu
_{20},\sigma_0})\right)   \\
&& \ \ \ > \frac{\lambda}{2} \sqrt{\frac{n_1 + n_2}{n_1n_2}} \left(\chi _{\alpha }^{2}(1) 
- \frac{2 n_1 n_2}{\lambda (n_1 + n_2)} d_{\gamma }(f_{\mu _{10},\sigma_0},f_{\mu
_{20},\sigma_0}) \right)\Bigg) \\
&= & 1 - \Phi_n\left( \frac{\lambda}{2\sigma_\gamma}\sqrt{\frac{n_1n_2}{n_1+n_2}} \left(\chi _{\alpha }^{2}(1) 
- \frac{2 n_1n_2}{n_1+n_2} d_{\gamma }(f_{\mu _{10},\sigma_0},f_{\mu
_{20},\sigma_0}) \right)\right),
\end{eqnarray*}
%
%
where $\Phi _{n}$ is a sequence of distributions functions tending uniformly
to the standard normal distribution function $\Phi$, and $\sigma_\gamma$ 
is defined in (\ref{sigma_gamma}). We observe that if $\mu _{10}\neq \mu _{20}$
%
%
\begin{equation}
 \lim_{n_1,n_2\rightarrow \infty }\eta_{\gamma,\beta}(\mu _{10}, \mu _{20}, \sigma_0) =1.
\end{equation}
%
%
Therefore, the test is consistent in the Frasar's sense \citep{MR0093863}. 

\begin{corollary}
\label{Cor1}Let $w\in (0,1)$ as defined in (\ref{EQ:w}) and $\gamma =\beta
=0$. Then, under the null hypothesis defined in (\ref{EQ:0}), we have
%
%
\begin{equation}
S_{0}\left( \widehat{\mu }_{10},\widehat{\mu }_{20},\widehat{\sigma }%
_{0}\right) =\frac{n_1n_2}{n_1+n_2}\frac{\left( \bar{X}-%
\bar{Y}\right) ^{2}}{\widehat{\sigma }_{0}^{2}}\underset{%
n_1,n_2\rightarrow \infty }{\overset{\mathcal{L}}{\longrightarrow }}\chi
^{2}(1).  \label{eqCor1}
\end{equation}
%
%
\end{corollary}

\noindent
The proof of the corollary is  straightforward. The test statistic given in the above corollary is closely related to the likelihood ratio test. This correspondence is described
in the next corollary.

\begin{corollary}
\label{Cor2}For a given sample the value of the test statistic $S_{0}\left(
\widehat{\mu }_{10},\widehat{\mu }_{20},\widehat{\sigma }_{0}\right) $, defined in (%
\ref{eqCor1}), does not exactly match the value of the likelihood ratio test statistic%
\begin{equation*}
-2\log \Lambda \left( \widehat{\mu }_{10},\widehat{\mu }_{20},\widehat{%
\sigma }_{0}\right) =(n_1+n_2)\log \left( 1+\frac{n_1n_2}{\left(
n_1+n_2\right) ^{2}}\frac{(\bar{X}-\bar{Y})^2}{%
\widehat{\sigma }_{0}^{2}}\right) ,
\end{equation*}%
where $\widehat{\sigma }_{0}^{2}$ is defined in (\ref{sig0}). However, as $%
n_1,n_2\rightarrow \infty $, and $w \in (0,1)$ as defined in (\ref{EQ:w}), both test statistics are asymptotically equivalent.
\end{corollary}

\begin{proof}
Let us denote  $\Theta _{0}=\left\{ \left( \mu ,\mu ,\sigma \right)^T :\mu
\in \mathbb{R},\sigma \in \mathbb{R}^{+}\right\} ,$ $\Theta =\left\{ \left( \mu_1,\mu_2,\sigma \right)^T
:\mu_1,\mu_2\in \mathbb{R},\sigma \in \mathbb{R}^{+}\right\}$. The likelihood function is given by
%
%
\begin{equation*}
\mathcal{L}(\mu_1,\mu_2,\sigma
)=\prod_{i=1}^{n_1}\prod_{j=1}^{n_2}f_{\mu_1,\sigma }(X_i)f_{\mu
_{2},\sigma }(Y_{j}).
\end{equation*}%
%
%
It can be shown that
%
%
\begin{equation*}
\Lambda \left( \widehat{\mu }_{10},\widehat{\mu }_{20},\widehat{\sigma }%
_{0}\right) =\frac{\sup_{\mu_1,\mu_2,\sigma \in \Theta _{0}}\mathcal{L}%
(\mu_1,\mu_2,\sigma )}{\sup_{\mu_1,\mu_2,\sigma \in \Theta }%
\mathcal{L}(\mu_1,\mu_2,\sigma )}=\left( \frac{\sum_{i=1}^{n_1}%
\left( X_i-\widetilde{\mu }\right) ^{2}+\sum_{i=1}^{n_2}\left( Y_{i}-%
\widetilde{\mu }\right) ^{2}}{\sum_{i=1}^{n_1}\left( X_i-\bar{X}%
\right) ^{2}+\sum_{i=1}^{n_2}\left( Y_{i}-\bar{Y}\right) ^{2}}\right)
^{-\frac{n_1+n_2}{2}},
\end{equation*}%
%
%
where $\widetilde{\mu }=\frac{n_1}{n_1+n_2}\bar{X}+\frac{n_2}{%
n_1+n_2}\bar{Y}$. %
Therefore, asymptotically, the likelihood ratio test rejects the null hypothesis $H_{0}$
if
%
%
\begin{equation*}
-2\log \Lambda \left( \widehat{\mu }_{10},\widehat{\mu }_{20},\widehat{%
\sigma }_{0}\right) =(n_1+n_2)\log \left( \frac{\sum_{i=1}^{n_1}\left(
X_i-\widetilde{\mu }\right) ^{2}+\sum_{i=1}^{n_2}\left( Y_{i}-\widetilde{%
\mu }\right) ^{2}}{\sum_{i=1}^{n_1}\left( X_i-\bar{X}\right)
^{2}+\sum_{i=1}^{n_2}\left( Y_{i}-\bar{Y}\right) ^{2}}\right) >\chi
^{2}(1).
\end{equation*}
%
%
Now
%
%
\begin{align*}
\sum_{i=1}^{n_1}\left( X_i-\widetilde{\mu }\right)
^{2}+\sum\limits_{i=1}^{n_2}\left( Y_{i}-\widetilde{\mu }\right) ^{2}&
=\sum\limits_{i=1}^{n_1}\left( X_i-\bar{X}\right) ^{2}+n_1\left(
\bar{X}-\widetilde{\mu }\right) ^{2}+\sum\limits_{i=1}^{n_2}\left(
Y_{i}-\bar{Y}\right) ^{2}+n_2\left( \bar{Y}-\widetilde{\mu }%
\right) ^{2} \\
& =\sum\limits_{i=1}^{n_1}\left( X_i-\bar{X}\right)
^{2}+\sum\limits_{i=1}^{n_2}\left( Y_{i}-\bar{Y}\right) ^{2}+\frac{%
n_1n_2}{n_1+n_2}(\bar{X}-\bar{Y})^{2}.
\end{align*}
%
%
So
%
%
\begin{eqnarray*}
-2\log \Lambda \left( \widehat{\mu }_{10},\widehat{\mu }_{20},\widehat{%
\sigma }_{0}\right) & = & (n_1+n_2)\log \left( \frac{\sum_{i=1}^{n_1}\left( X_i-\widetilde{\mu }\right)
^{2}+\sum_{i=1}^{n_2}\left( Y_{i}-\widetilde{\mu }\right) ^{2}}{%
\sum_{i=1}^{n_1}\left( X_i-\bar{X}\right)
^{2}+\sum_{i=1}^{n_2}\left( Y_{i}-\bar{Y}\right) ^{2}}\right)
\\
& =& (n_1+n_2)\log \left( 1+\frac{n_1n_2}{\left( n_1+n_2\right)
^{2}}\frac{(\bar{X}-\bar{Y})^{2}}{\widehat{\sigma }_{0}^{2}}%
\right)
\\
& = &\displaystyle \frac{n_1n_2}{\left( n_1+n_2\right)
}\frac{(\bar{X}-\bar{Y})^{2}}{\widehat{\sigma }_{0}^{2}} + R_{n_1,n_2},
\end{eqnarray*}%
%
%
where $R_{n_1,n_2} \rightarrow 0$ in probability as  $n_1,n_2\rightarrow \infty$ and $w \in (0,1)$.
Thus, the test statistics  $-2\log \Lambda \left( \widehat{\mu }_{10},\widehat{\mu }%
_{20},\widehat{\sigma }_{0}\right)$ and $S_{0}\left( \widehat{\mu }_{10},%
\widehat{\mu }_{20},\widehat{\sigma }_{0}\right)$ are asymptotically equivalent.
\hspace*{\fill}%
\end{proof}

\section{Numerical Studies}\label{SEC:numerical}

\subsection{Simulation Study}
In this section we study the performance of our proposed test statistics
through simulated data. We have generated two random samples $%
X_{1},X_{2},\ldots ,X_{n_1}$ and $Y_{1},Y_{2},\ldots ,Y_{n_2}$ from $%
\mathcal{N}(\mu_1,\sigma^2)$ and $\mathcal{N}(\mu_2,\sigma^2)$
respectively; thus the total sample size is $n=n_1+n_2$. The value of $w$ in (%
\ref{EQ:w}) is taken to be 0.6, and the sample size from the first
population is $n_1=[wn]+1$, where $[x]$ denotes the integer part of $x$.
Our aim is to test the null hypothesis given in (\ref{EQ:0}). We have taken $%
\sigma^2=1$ in this study. We have compared the results of the ordinary
two sample $t$-test and the density power divergence tests with four
different values of the tuning parameter $\gamma =\beta =0,0.05, 0.1$ and 0.15;
let DPD($\gamma $) represent the DPD test with tuning parameter $%
\gamma $. The nominal level of the tests are 0.05, and all tests are
replicated 1,000 times.

In the first case we have taken $\mu_{1}=\mu_{2}=0$. Plot (a) in Figure %
\ref{fig:level_power} shows the observed levels of the five test statistics
for different values of the sample size (obtained as the proportion of test
statistics, in the $1000$ replications, that exceed the nominal $\chi ^{2}$
critical value at 5\% level of significance). It is seen that the observed
levels of the $t$-test are very close to the nominal level. On the other
hand, the DPD tests are slightly liberal for very small sample sizes and
lead to somewhat inflated observed levels. However, as the sample size
increases the levels settle down rapidly around the nominal level.

Next, we have generated data with $\mu_1 = 0$ but $\mu_2 = 1$. The observed power
of the tests are presented in plot (b) of Figure \ref{fig:level_power}.
There is not much difference among the observed powers in this plot. The DPD
tests have slightly higher power than the $t$-test in very small sample
sizes. This, however, must be a consequence of the fact that the observed
levels of these tests are higher than the nominal level (and higher than the
observed level of the $t$-statistic) in small samples.

Now we check the performance of the tests under contaminated data. So, we have
generated $n_2$ observations $Y_{1},Y_{2},\ldots ,Y_{n_2}$ from $0.95%
\mathcal{N}(\mu_2,1)+0.05\mathcal{N}(-10,1)$, whereas the $n_1$
observations representing the first population come from the pure $\mathcal{N%
}(\mu_1,1)$ distribution. To evaluate the stability of the level of the tests
for testing the hypothesis in (\ref{EQ:0}), we have taken $\mu_1=\mu
_{2}=0 $. Figure \ref{fig:level_power} (c) presents the levels for different
values of the sample sizes. It may be observed that there is a drastic
inflation in the levels for the $t$-test and DPD(0) test statistic, but the
levels of the other DPD test statistics remain stable.

Figure \ref{fig:level_power} (d) shows the power of the tests under the
contaminated setup considered in the previous paragraph, when $\mu_1 = 0$
and $\mu_2 = 1$. Here, the presence of the outliers lead to a sharp drop in
power for the $t$-test and the DPD(0) test. On the other hand, the other
tests are clearly more resistant, and hold their power much better as $\gamma$
increases.

On the whole, therefore, it appears that in comparison to the $t$-test, many
of our DPD tests are quite competitive in performance when the data come
from the pure model. Under contaminated data, however, the robustness
properties of the DPD tests appear to be far superior, and they do much
better at maintaining the stability of the level and the power in such cases.

%
%
\begin{figure}
\centering%
\begin{tabular}{rl}
\includegraphics[height=7.5cm, width=7.5cm]{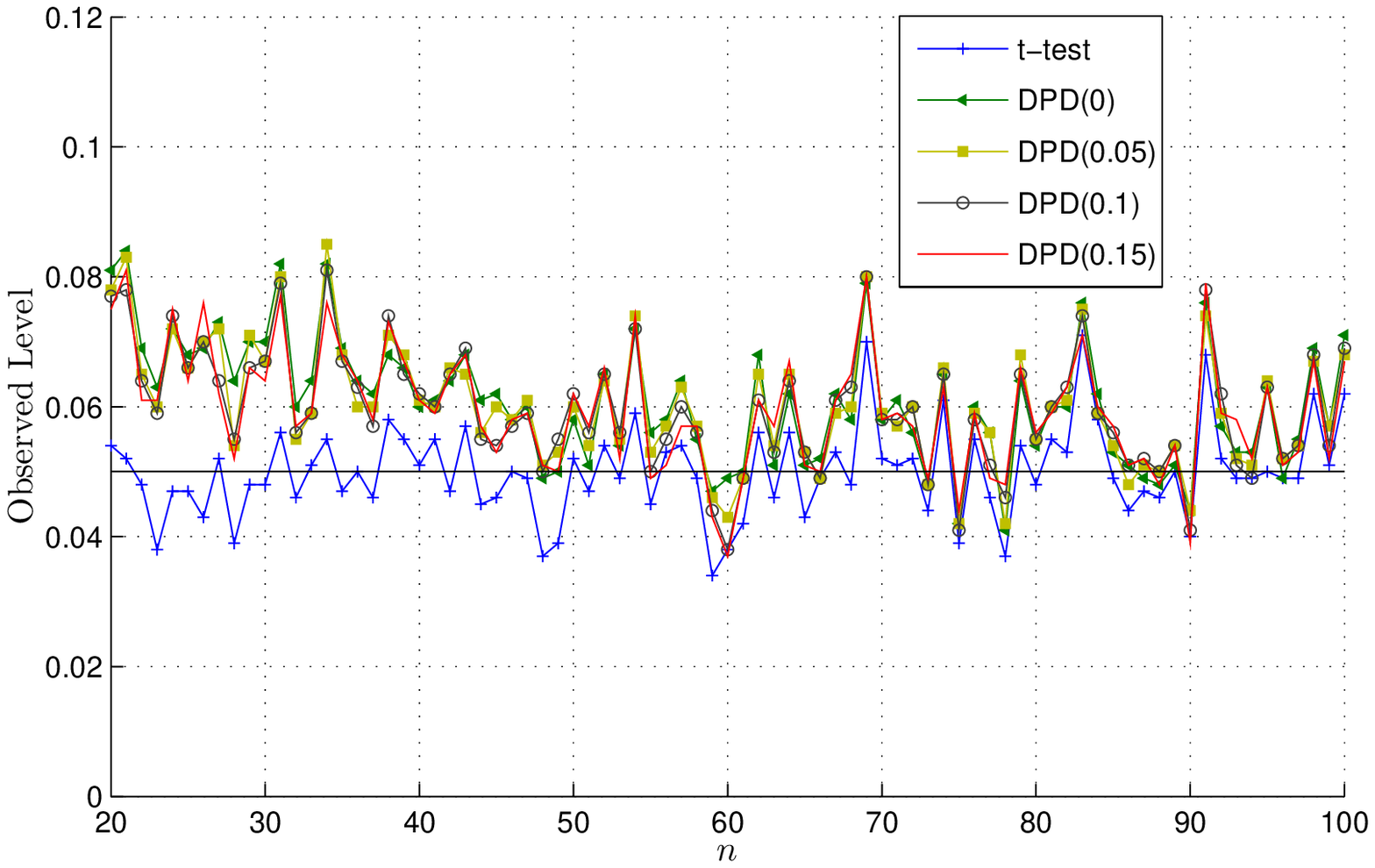}\negthinspace &
\negthinspace \includegraphics[height=7.5cm, width=7.5cm]{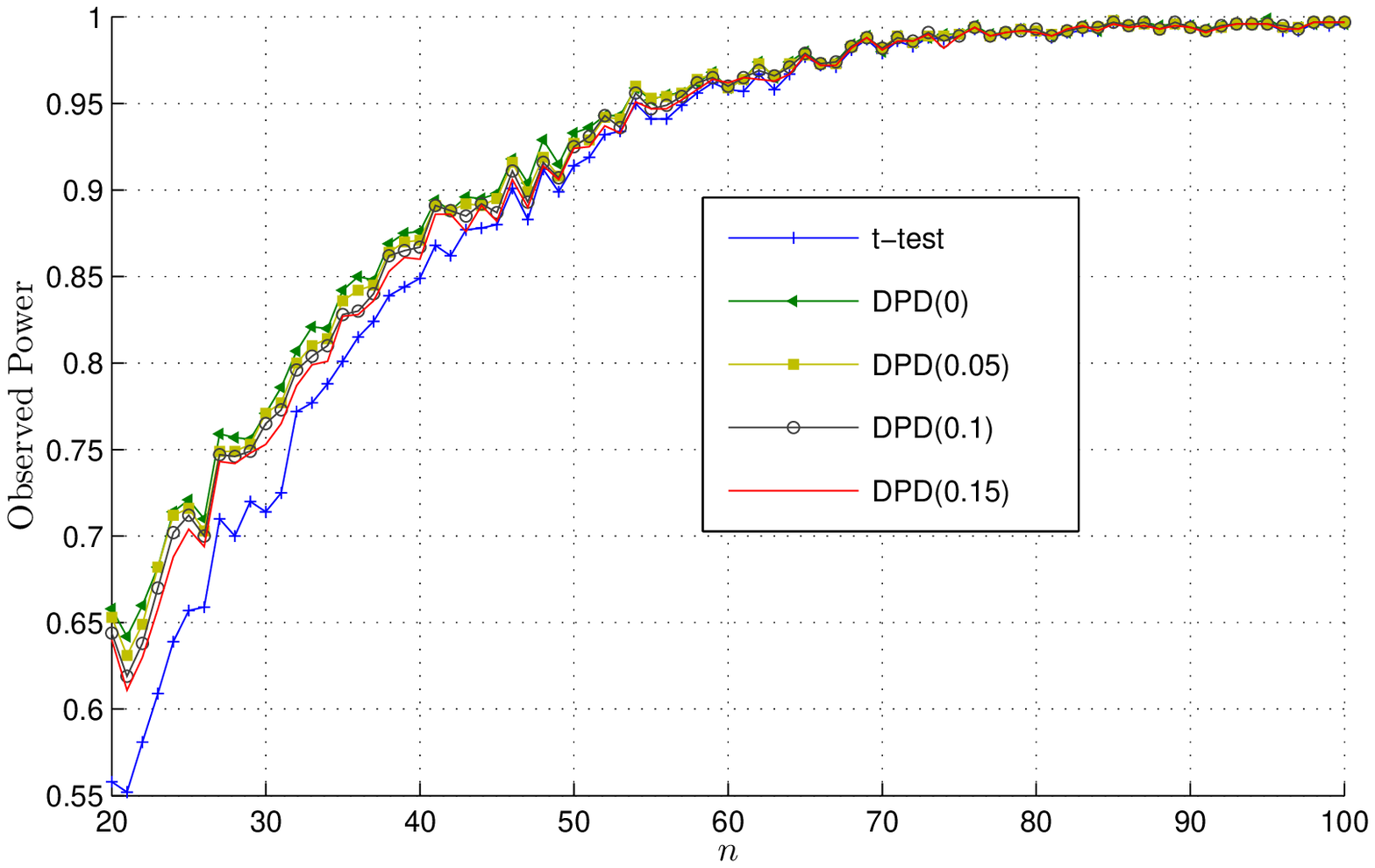} \\
\multicolumn{1}{c}{\textbf{(a)}} & \multicolumn{1}{c}{\textbf{(b)}} \\
\includegraphics[height=7.5cm, width=7.5cm]{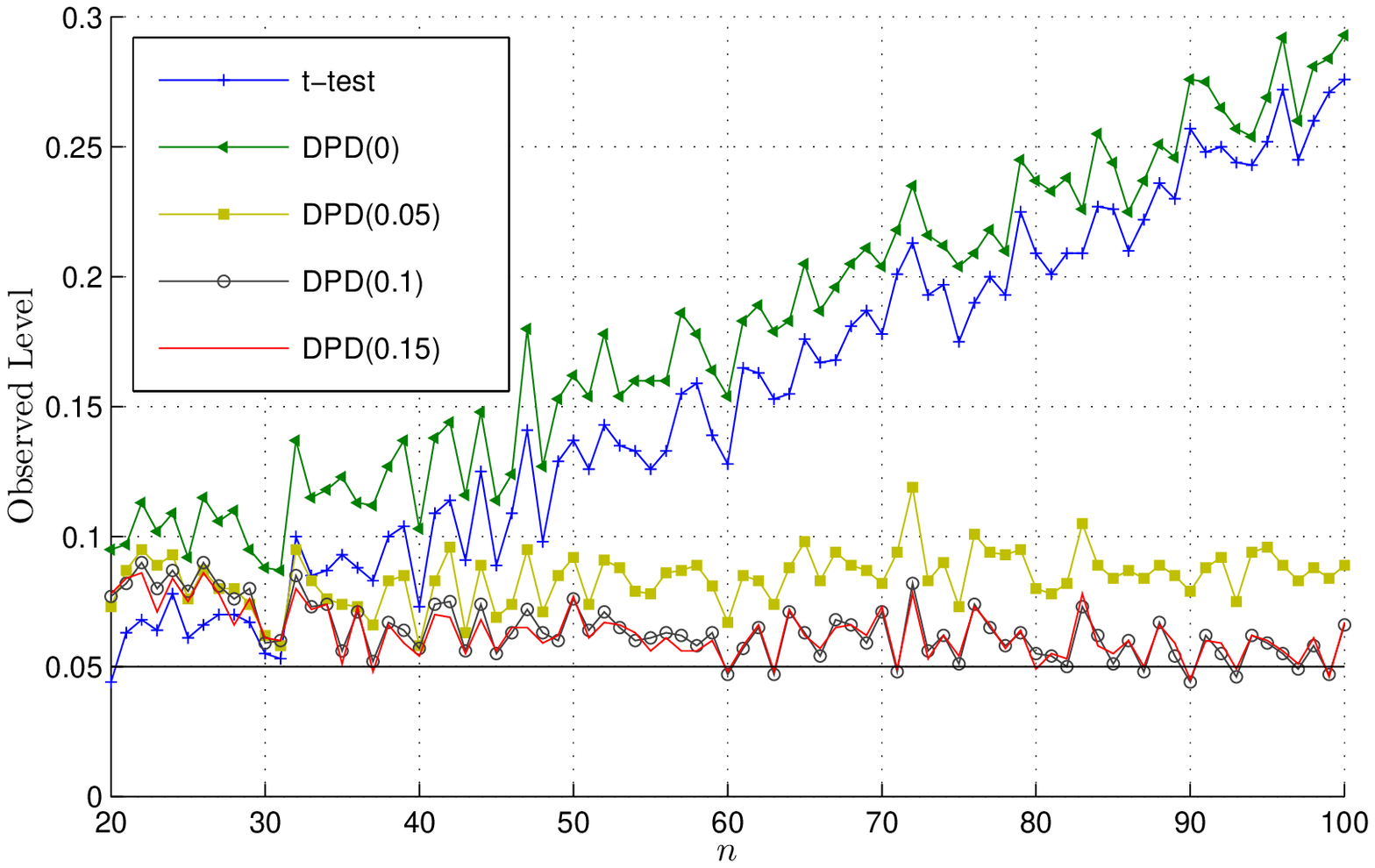}\negthinspace &
\negthinspace \includegraphics[height=7.5cm, width=7.5cm]{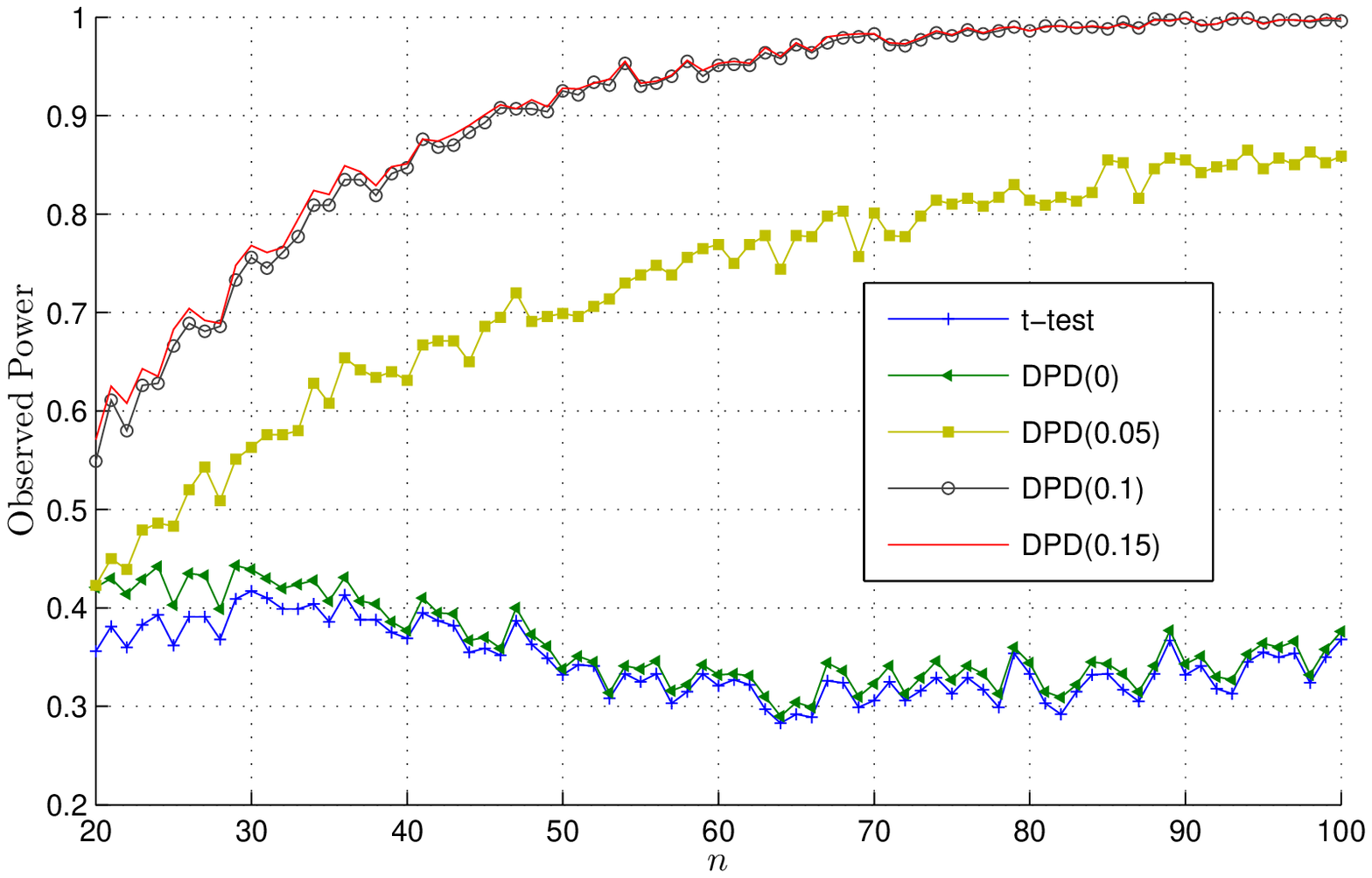} \\
\multicolumn{1}{c}{\textbf{(c)}} & \multicolumn{1}{c}{\textbf{(d)}}%
\end{tabular}%
\caption{(a) Simulated levels of the DPD tests for pure data; (b) simulated power of the DPD tests for pure data;
(c) simulated levels of the DPD tests for contaminated data; (d) simulated power of the DPD tests for contaminated data.}
\label{fig:level_power}
\end{figure}
%
%

\subsection{Comparison with Other Robust Tests}
In this section we compare the DPD test with some other popular robust tests. For comparison we have used a parametric test -- the two sample trimmed $t$-test proposed by \cite{yuen1973approximate}, as well as two non-parametric tests -- the Kolmogorov-Smirnov test (KS-test) and the Wilcoxon two-sample test (which is also known as the Mann-Whitney $U$-test). For the two sample trimmed $t$-test we have trimmed 20\% extreme observations from each of the data sets of $X$ and $Y$. The set up, the parameters taken for the simulation and the level of contamination are exactly the same as in the previous section. For comparison we have used only one DPD test in this case, that corresponding to tuning parameter 0.1. To emphasize the robustness properties of these tests we have also included the two sample $t$-test in this investigation. The results are presented in Figure \ref{fig:level_power_v1}. 

Figure \ref{fig:level_power_v1} (a) shows that the observed levels of all the robust tests are very close to the nominal level of 0.05 for the pure normal data. The same result is observed in Figure \ref{fig:level_power_v1} (c) for the contaminated data. On the other hand, if we consider the observed power of the tests the DPD test is much more powerful than the other tests. Specifically, for the contaminated data, the DPD test does significantly better than the others in holding on to its power. Therefore, on the whole, the DPD tests are not only superior to the two sample $t$-test under contamination, but they also appear to be competitive or better than the other popular robust tests as far as this simulation study is concerned.

\begin{figure}
\centering%
\begin{tabular}{rl}
\includegraphics[height=7.5cm, width=7.5cm]{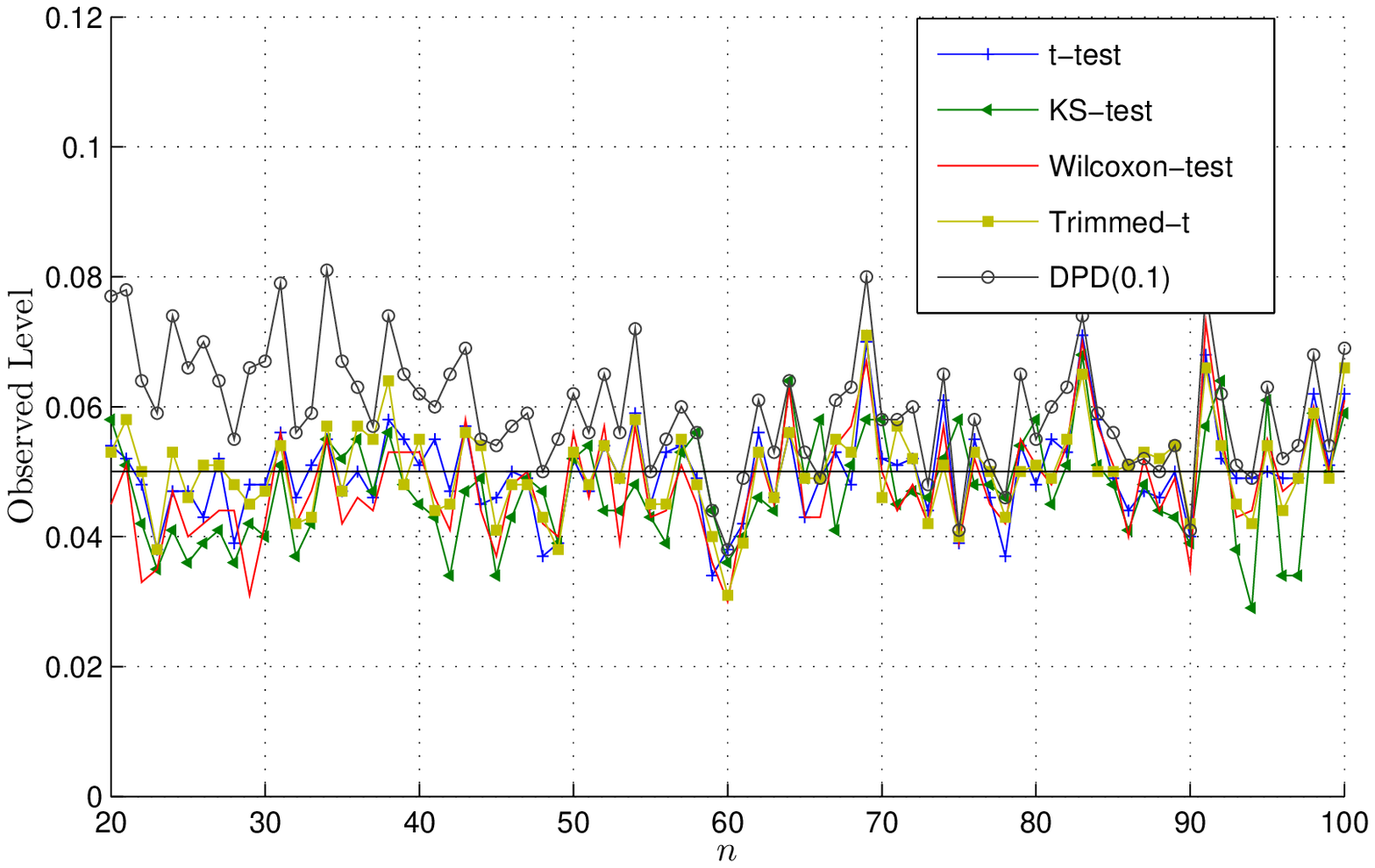}\negthinspace &
\negthinspace \includegraphics[height=7.5cm, width=7.5cm]{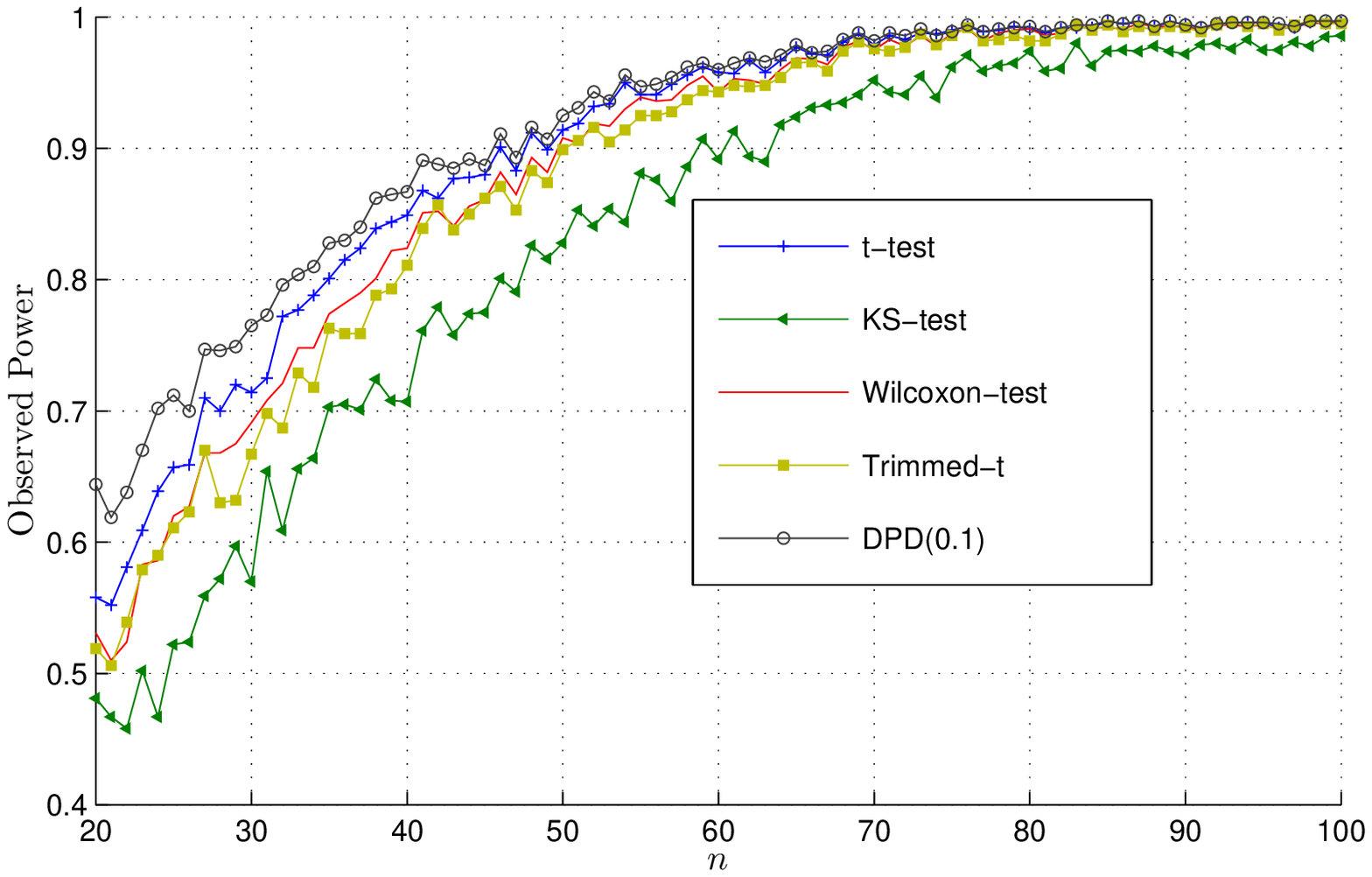} \\
\multicolumn{1}{c}{\textbf{(a)}} & \multicolumn{1}{c}{\textbf{(b)}} \\
\includegraphics[height=7.5cm, width=7.5cm]{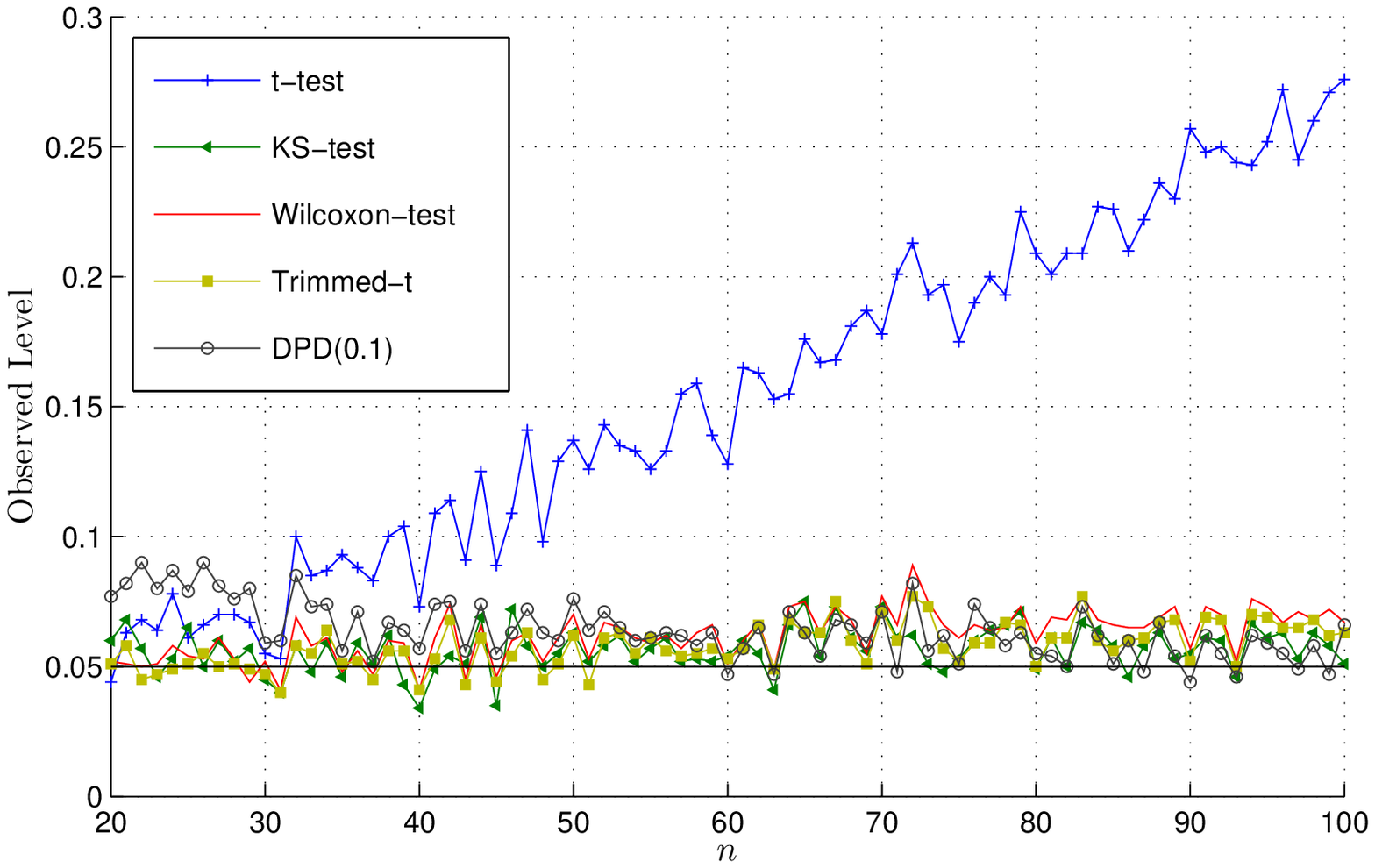}\negthinspace &
\negthinspace \includegraphics[height=7.5cm, width=7.5cm]{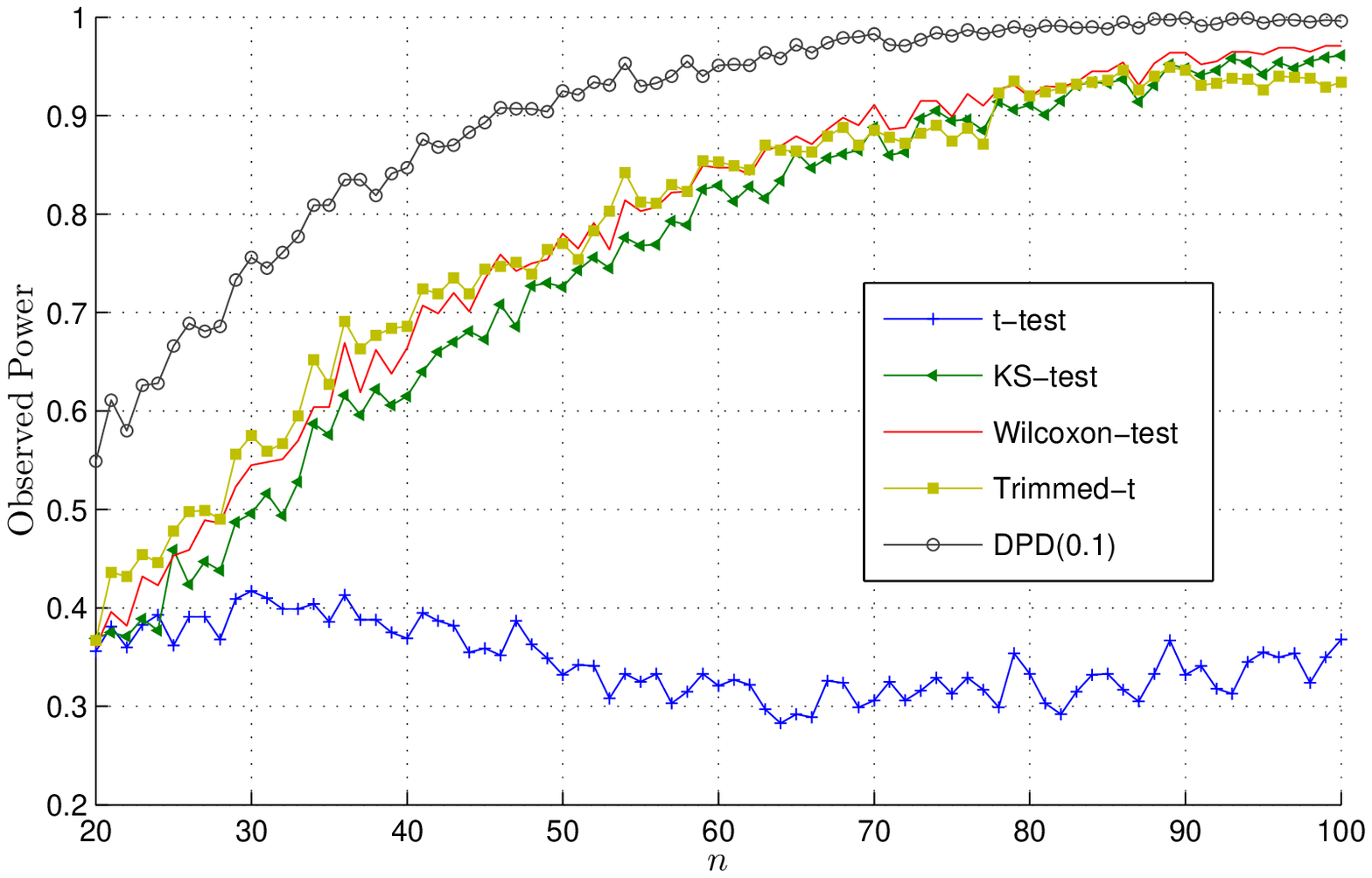} \\
\multicolumn{1}{c}{\textbf{(c)}} & \multicolumn{1}{c}{\textbf{(d)}}%
\end{tabular}
\caption{(a) Simulated levels of different tests for pure data; (b) simulated power of different tests for pure data;
(c) simulated levels of different tests for contaminated data; (d) simulated power of different tests for contaminated data.}
\label{fig:level_power_v1}
\end{figure}
\subsection{Real Data Examples}

\noindent \textbf{Example 2 (Lead Measurement data):} %
In Table \ref{TAB:Lake} the lead measurement data (\citealp{MR922042}, p. 280)
are presented. The numbers represent the values of $10(x-2)$, $x$
being the level of lead in the water samples from two lakes at randomly
chosen locations. To test whether the average pollution levels of the two
lakes are equal, we perform tests for equality of the means of the
populations represented by the two different samples. The $p$-values of the
DPD tests are plotted in Figure \ref{fig:Lake_data_p_val}; the solid line
represents the $p$-values for the full data, while the dashed line
represents for the $p$-values for the outlier deleted data. The less robust
tests (corresponding to very small values of $\gamma$) register only borderline 
significance under full data, and for very small values of $\gamma$ the tests 
would fail to reject the equality hypothesis at the 1\% significance level. 
However, for all value of $\gamma$, the tests would soundly reject the null hypothesis when the obvious
outliers (displayed with bold fonts in Table \ref{TAB:Lake}) are removed from
the dataset. For higher values of $\gamma$ (0.2 or larger), the $p$-values
with or without the outliers are practically identical, demonstrating that
the outliers have little effect in such cases. The $p$-values for the
two-sample $t$-test with and without the outliers are 0.02397 and 0.0004
respectively. As in Example 1, the presence of the outliers masks the
significance of the two-sample $t$-test and the small $\gamma$ DPD tests,
but the large $\gamma$ DPD tests successfully discount the effect of the
outliers.
%
%
\begin{table}
\caption{Lead Measurement data.}
\label{TAB:Lake}
\begin{center}
\begin{tabular}{lrrrrrrrrrr}
\hline
First Lake & $-1.48$ & $1.25$ & $-0.51$ & $0.46$ & $0.60$ & ${\bf -4.27}$ & $0.63$
& $-0.14$ & $-0.38$ & $1.28$ \\
& $0.93$ & $0.51$ & $1.11$ & $-0.17$ & $-0.79$ & $-1.02$ & $-0.91$ & $0.10$
& $0.41$ & $1.11$ \\ \hline
Second Lake & $1.32$ & $1.81$ & $-0.54$ & $2.68$ & $2.27$ & $2.70$ & $0.78$
& ${\bf -4.62}$ & $1.88$ & $0.86$ \\
& $2.86$ & $0.47$ & $-0.42$ & $0.16$ & $0.69$ & $0.78$ & $1.72$ & $1.57$ & $%
2.14$ & $1.62$ \\ \hline
\end{tabular}%
\end{center}
\end{table}
%
%
%
\begin{figure}
\centering{\includegraphics[height=6.5cm, width=14cm]{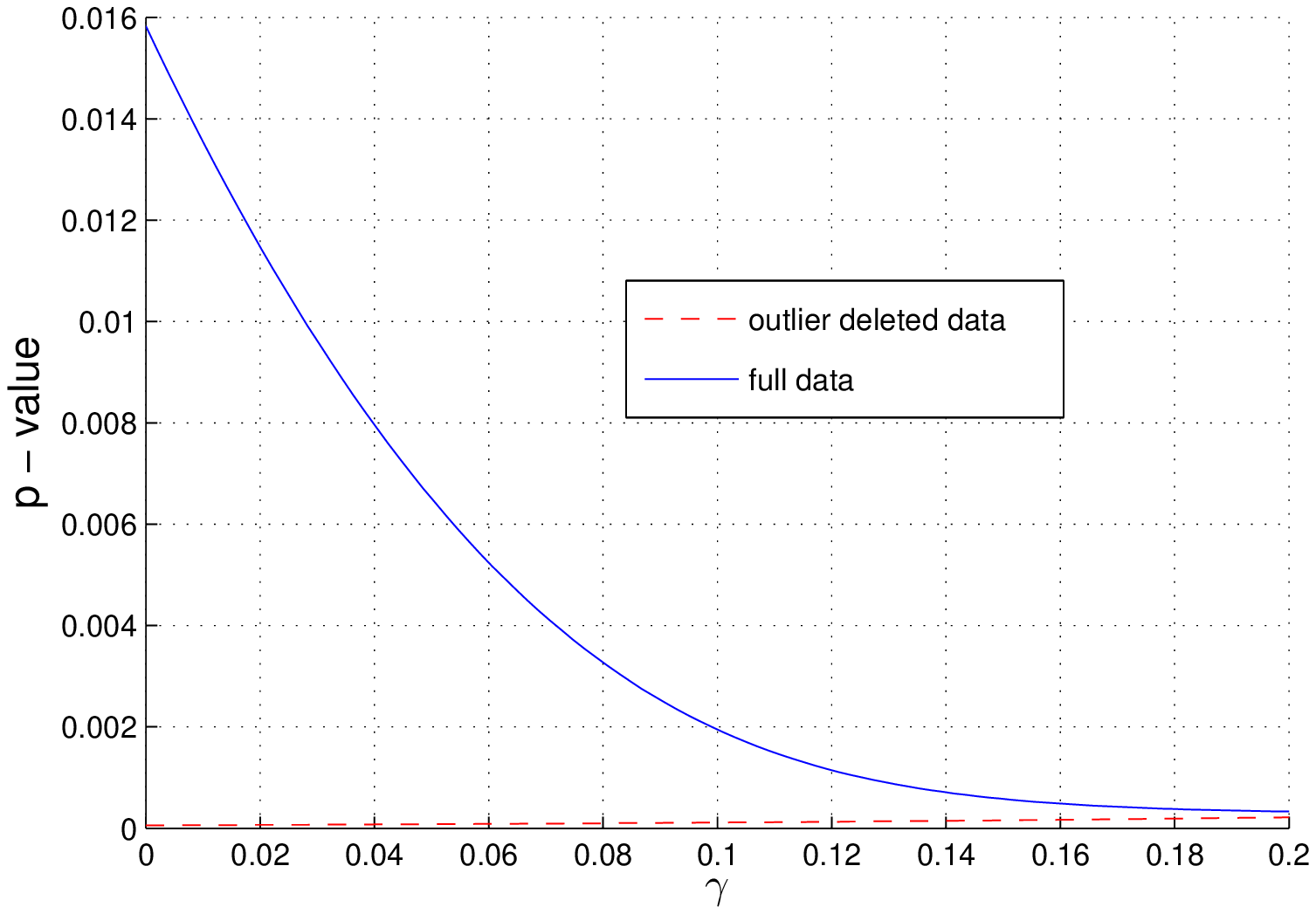}}
\caption{The $p$-values of the DPD tests for the Lead Measurement data for different values of $\gamma$. The solid line represents the full data analysis, while the dashed line represents the outlier deleted case.}
\label{fig:Lake_data_p_val}
\end{figure}
%

\bigskip \noindent \textbf{Example 3 (Ozone Control data):} \cite{MR0443210}
report data from a study design to assess the effects of ozone on weight
gain in rats. The experimental group consisted of 22 rats, each 70-day old kept in
an ozone environment for 7 days. A control group of 23 rats, of the same
age, were kept in an ozone-free environment. The weight gains, in grams, are
listed in Table \ref{TAB:Ozone}. We want to test for the equality of the
means of the two groups. The $p$-values of the DPD tests are plotted in
Figure \ref{fig:Ozone_control_data_p_val}. The $p$-values of the two-sample $%
t$-test for the full data and the outlier deleted data are $0.0168$ and $%
3.4721\times 10^{-6}$ respectively. The conclusions of this example are
similar to those of Examples 1 and 2. 

%
%
\begin{table}
\caption{Ozone Control data}
\label{TAB:Ozone}
\begin{center}
\begin{tabular}{lrrrrrrrrrrrr}
\hline
$X$ & $\bf{41.0}$ & $\bf{38.4}$ & $24.4$ & $25.9$ & $21.9$ & $18.3$ & $13.1$ & $27.3$
& $28.5$ & $\bf{-16.9}$ & $26.0$ & $17.4$ \\
& $21.8$ & $15.4$ & $27.4$ & $19.2$ & $22.4$ & $17.7$ & $26$ & $29.4$ & $%
21.4 $ & $26.6$ & $22.7$ &  \\ \hline
$Y$ & $10.1$ & $6.1$ & $20.4$ & $7.3$ & $14.3$ & $15.5$ & $-9.9$ & $6.8$ & $%
28.2$ & $17.9$ & $-9.0$ & $-12.9$ \\
& $14.0$ & $6.6$ & $12.1$ & $15.7$ & $\bf{39.9}$ & $-15.9$ & $\bf{54.6}$ & $-14.7$ & $%
\bf{44.1}$ & $-9.0$ &  &  \\ \hline
\end{tabular}%
\end{center}
\end{table}
%
%

%
%
\begin{figure}
\centering{%
\includegraphics[height=6.5cm,
width=14cm]{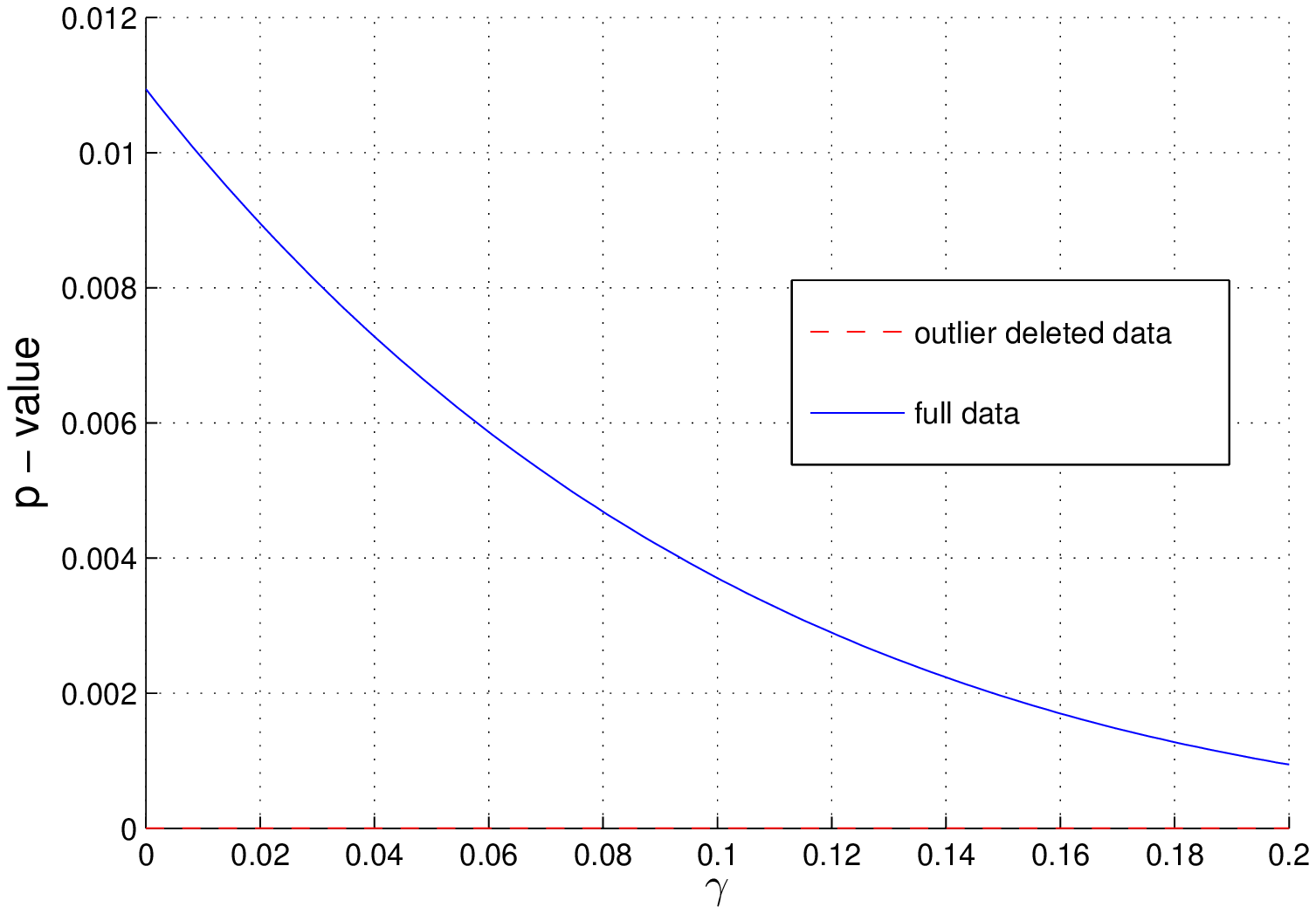}}
\caption{The $p$-values of the DPD tests for the Ozone Control data for different values of $\gamma$. The solid line represents the full data analysis, while the dashed line represents the outlier deleted case.}
\label{fig:Ozone_control_data_p_val}
\end{figure}
%
%

\bigskip \noindent \textbf{Example 4 (Newcomb's Light Speed data)}: In 1882 Simon
Newcomb, an astronomer and mathematician, measured the time required for a
light signal to pass from his laboratory on the Potomac River to a mirror at
the base of the Washington Monument and back. The total distance was $%
7443.73 $ meters. Table \ref{TAB:tNewcomb} contains these measurements from
three samples, as deviations from $24,800$ nanoseconds. For example, for the
first observation, $28$, means that the time taken for the light to travel
the required $7443.73$ meters is $24,828$ nanoseconds. The data comprises
three samples, of sizes $20$, $20$ and $26$, respectively, corresponding to
three different days. These data have been analyzed previously by a number
of authors including \cite{MR0455205} and \cite{Voinov}. The $p$-values
of the DPD statistics for the test of the equality of means between Day 1
and Day 2, and Day 1 and Day 3 are plotted in Figure \ref{fig:Newcomb_data12_p_val}, and \ref{fig:Newcomb_data13_p_val} respectively.

The $p$-values for the two-sample $t$-tests for the (Day 1, Day 2)
comparison are $0.1058$ for the full data case, and $0.3091$ for the outlier
deleted case. The same for the (Day 1, Day 3) comparison are $0.0970$ and $%
0.2895$ respectively. However, for the large $\gamma $, the results from the  DPD tests
are clearly insignificant with or without the outliers. In this example,
therefore, the outliers are forcing the outcome of the two-sample $t$-test
(and the DPD tests for small $\gamma $) to the borderline of significance,
but the robust tests give insignificant results with or without the
outliers, preventing the false significance that is produced by the outliers
in the $t$-test; this is unlike the previous three examples where the robust
tests overcame a masking effect. These examples demonstrate that the robust
DPD tests can give protection against spurious conclusions in both
directions.

\begin{table}[tbp]
\caption{Newcomb's Light Speed data.}
\label{TAB:tNewcomb}
\begin{center}
\begin{tabular}{lrrrrrrrrrrrrrrr}
\hline
day 1 & $28$ & $26$ & $33$ & $24$ & $34$ & $\bf{-44}$ & $27$ & $16$ & $40$ & $\bf{-2}$
& $29$ & $22$ & $24$ & $21$ & $25$ \\
& $30$ & $23$ & $29$ & $31$ & $19$ &  &  &  &  &  &  &  &  &  &  \\ \hline
day 2 & $24$ & $20$ & $36$ & $32$ & $36$ & $28$ & $25$ & $21$ & $28$ & $29$
& $37$ & $25$ & $28$ & $26$ & $30$ \\
& $32$ & $36$ & $26$ & $30$ & $22$ &  &  &  &  &  &  &  &  &  &  \\ \hline
day 3 & $36$ & $23$ & $27$ & $27$ & $28$ & $27$ & $31$ & $27$ & $26$ & $33$
& $26$ & $32$ & $32$ & $24$ & $39$ \\
& $28$ & $24$ & $25$ & $32$ & $25$ & $29$ & $27$ & $28$ & $29$ & $16$ & $23$
&  &  &  &  \\ \hline
\end{tabular}%
\end{center}
\end{table}
%
%

%
%

%
%
\begin{figure}
\centering{%
\includegraphics[height=6.5cm,
width=14cm]{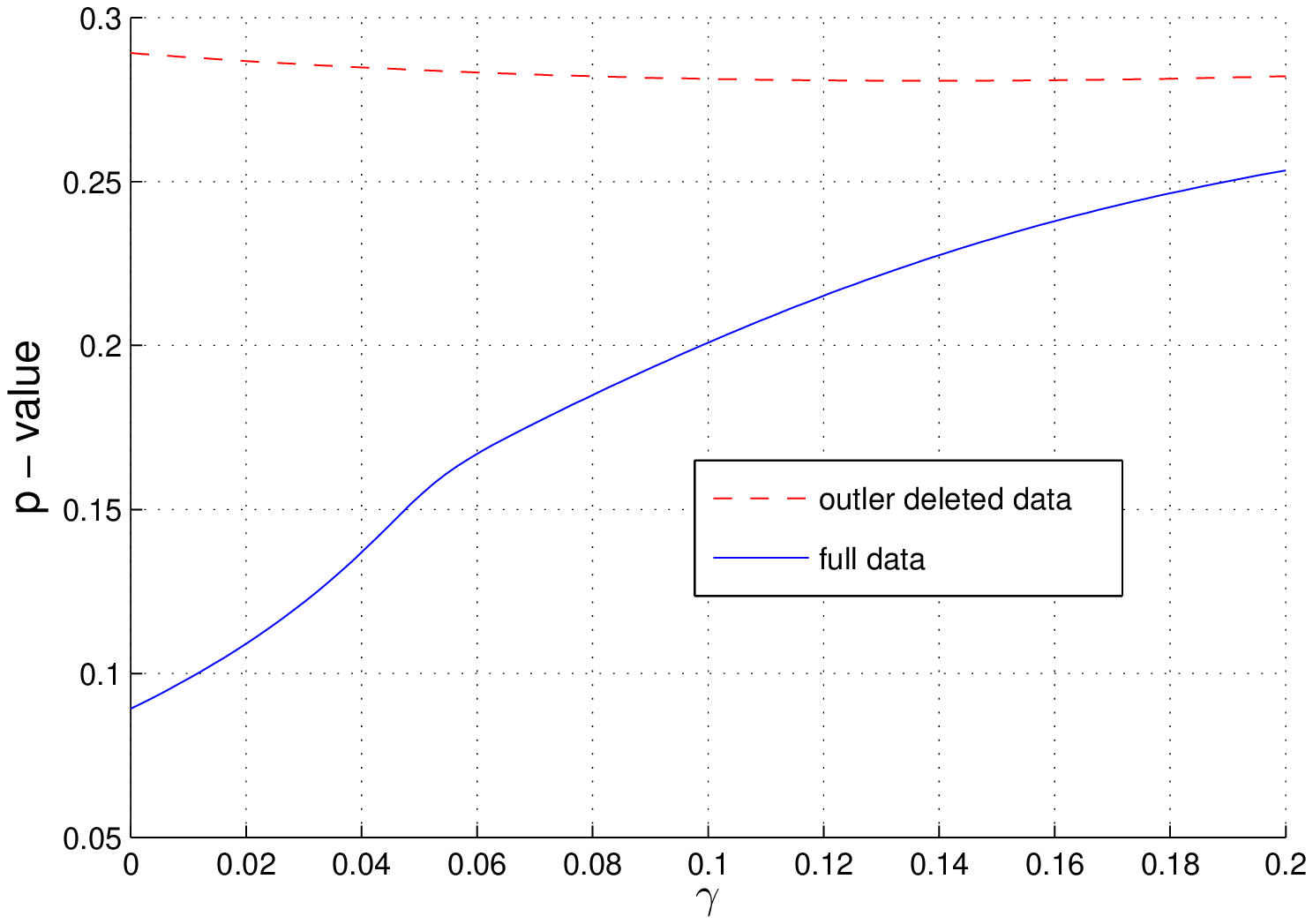}}
\caption{The $p$-values of the DPD tests for Newcomb's Light Speed data (Day 1 versus Day 2) for different values of $\gamma$. The solid line represents the full data analysis, while the dashed line represents the outlier deleted case.}
\label{fig:Newcomb_data12_p_val}
\end{figure}
\begin{figure}
\centering{%
\includegraphics[height=6.5cm,
width=14cm]{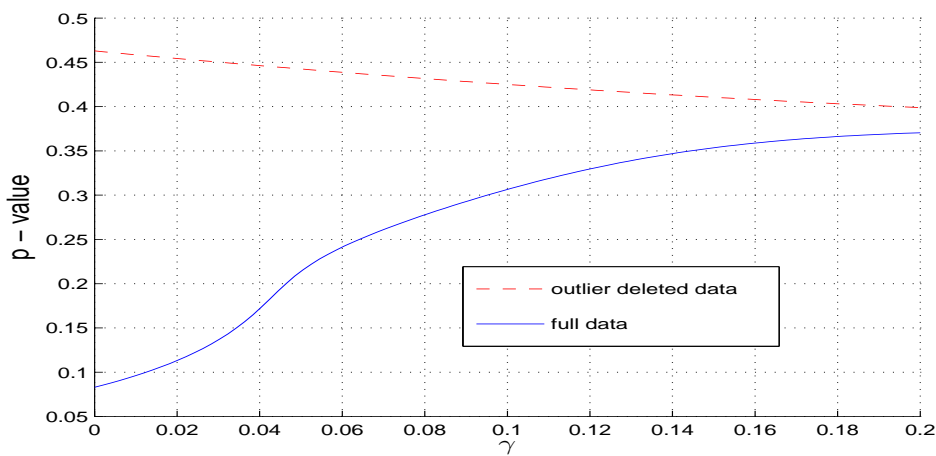}}
\caption{The $p$-values of the DPD tests for Newcomb's Light Speed data (Day 1 versus Day 3) for different values of $\gamma$. The solid line represents the full data analysis, while the dashed line represents the outlier deleted case.}
\label{fig:Newcomb_data13_p_val}
\end{figure}

\bigskip \noindent \textbf{Example 5 (Na Intake data)}: Sodium chloride preference was
determined in ten patients with essential hypertension and in 12 normal
volunteers. All exhibited normal detection and recognition thresholds for
the taste of sodium chloride. All were placed on a constant dry diet
containing 9 mEq of Na+ and given, as their only source of fluids, a choice
of drinking either distilled water or 0.15 M sodium chloride. Patients with
essential hypertension consumed a markedly greater proportion of their total
fluid intake as saline (38.2\% vs 10.6\%, average daily preference over one
week) and also showed a greater total fluid intake (1,269 ml vs 668 ml,
average daily intake over one week). The hypertensive patients consumed more
than four times as much salt as did the normal volunteers. The data are
given in Table \ref{TAB:Na_intake}. The $p$-values of the tests for the equality 
of means are plotted in Figure \ref{fig:Na_intake_data_p_val}. The findings are 
similar to examples 1, 2 and 3. 

\begin{table}
\caption{Na Intake data.}
\label{TAB:Na_intake}
\begin{center}
\begin{tabular}{lrrrrrrrrrrr}
\hline
$X$ & $114.6$ & $64.6$ & $70.4$ & $61.2$ & $\bf{297}$ & $60.9$ & $73.7$ & $15.7$
& $53.3$ &  &  \\
$Y$ & $14.2$ & $3.2$ & $3.7$ & $0.0$ & $73.6$ & $56.6$ & $97.2$ & $2.4$ & $%
0.0$ & $4.8$ & $0$ \\ \hline
\end{tabular}%
\end{center}
\end{table}
\begin{figure}
\centering{%
\includegraphics[height=6.5cm,
width=14cm]{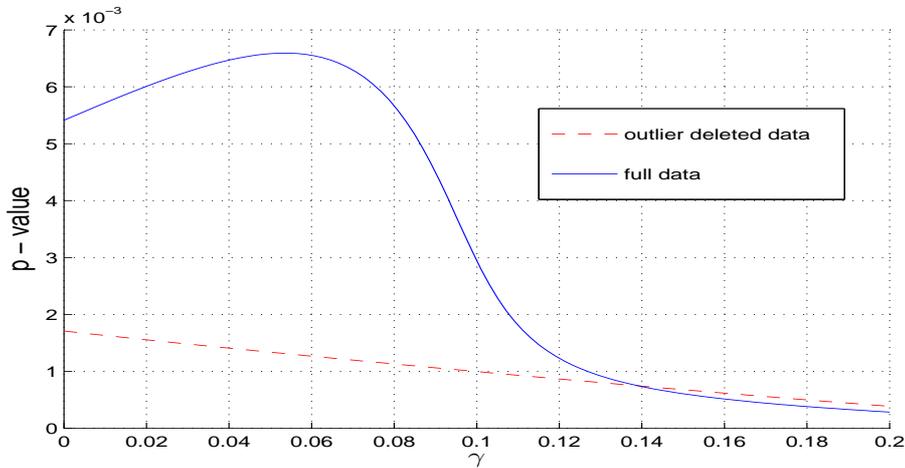}}
\caption{The $p$-values of the DPD tests for Na Intake for different values of $\gamma$. The solid line represents the full data analysis, while the dashed line represents the outlier deleted case.}
\label{fig:Na_intake_data_p_val}
\end{figure}
\bigskip

\noindent {\bf Example 6 (Sri Lanka Zinc Content data)}:
The impact of a polluted environment on the health of the residents of an 
area is a common environmental concern. Large amounts of heavy metals in 
the body may signal a serious health threat to a community. One study, 
performed in Sri Lanka, sought to compare rural Sri Lankans with their 
urban counterparts in terms of the zinc content of their hair. A collection of individuals from 
rural Sri Lanka was recruited, samples of their hair were taken, and the 
zinc content in the hair was measured. An independent collection of students 
from an urban environment was studied, with the zinc content in samples of their 
hair being measured as well. The data are given in Table \ref{TAB:SriLanka}. 
The $p$-values of the tests for the equality of the means are plotted in Figure 
\ref{fig:SriLanka}. The results again
indicate that the presence of outliers can mask the true significance in
case of the two sample $t$-test and DPD tests for small values of $\gamma$, but
for the large $\gamma$ DPD tests are much more stable in such situations.

\begin{table}
\caption{Sri Lanka Zinc Content data.}
\label{TAB:SriLanka}
\begin{center}
 \begin{tabular}{lccccccccccc}\hline
 Urban ($X$) & 1120 & 230 &  \bf{4200} &  1200 &  1400 &  750 &  2101 &  430 &  690 &  600 &  834 \\
Rural ($Y$) &  \bf{3619} &  1104 &  243 &  658 &  673 &  598 &  648 &  918 &  133 &  289 &  250 \\
            & 304 &     555 &  640 &  933 &      &      &      &      &      &      &        \\
\hline
\end{tabular}%
\end{center}
\end{table}
\begin{figure}
\centering
{\includegraphics[height=6.5cm,width=14cm]{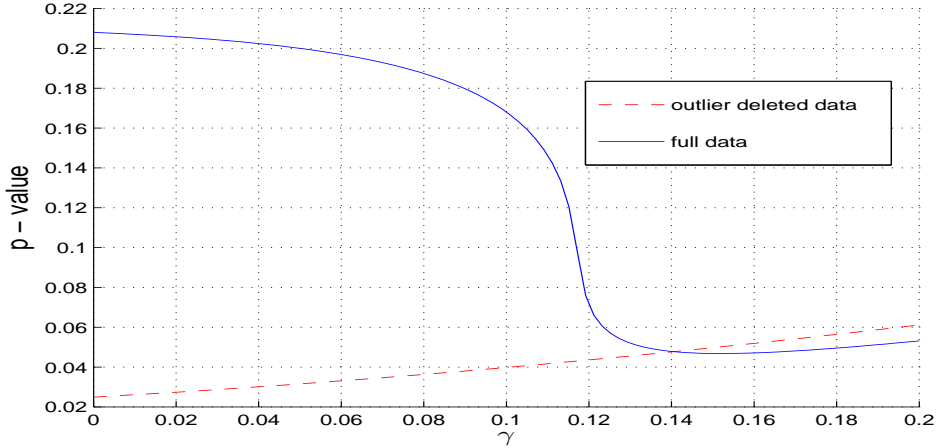}}
\caption{The $p$-values of the DPD tests for Sri Lanka Zinc Content data for different values of $\gamma$. The solid line represents the full data analysis, while the dashed line represents the outlier deleted case.}
\label{fig:SriLanka}
\end{figure}
%
%

\section{Concluding Remarks}\label{SEC:concluding}

Without any doubt, the two sample $t$-test is one of the most frequently used tools in the
statistics literature. It allows the experimenter to perform tests of the comparative hypotheses, 
which are the default requirements to be passed before one may declare that a new drug or treatment is 
an improvement over an existing one. The two sample $t$-test is simple to implement and has 
several optimality properties. In spite of such desirable attributes, this test is deficient on 
one count, which is that it does not retain its desired properties under contamination and model misspecification. As few 
as one, single, large outlier can turn around the decision of the test, and can make the 
resulting inference meaningless. In this paper we have introduced a test based on the density 
power divergence; the theoretical properties of the test have been rigorously determined. 
More importantly, we have demonstrated, through several real
data examples, that the DPD test is capable of uncovering 
both kinds of masking effects caused by outliers -- blurring the
true difference when one exists, and detecting a difference when
there is actually none.
 The test is simple to use and easy to understand, and we 
trust that it has the potential to become a powerful tool for the applied statistician.

\bigskip

\noindent
\textbf{Acknowledgments }This work was partially supported by Grants MTM-2012-33740 and ECO-2011-25706. The authors gratefully acknowledge the suggestions of  two anonymous referees which led to an improved version of the paper.

\bibliographystyle{abbrvnat}
\bibliography{reference}


\bigskip

\section*{Appendix}

\noindent \textbf{Proof of Theorem \ref{Th0}:} 
As $\widehat{\mu }_{i\beta }$ is the solution of the estimating equation $_{1}h_{n_{i},\beta }^{\prime }\left( \mu _{i},\sigma \right) =0$,
we get from equation (\ref{muSigma})
\begin{equation*}
\sqrt{n_{i}}(\widehat{\mu }_{i\beta }-\mu _{i0})=\sqrt{n_{i}} \boldsymbol{J}_{11,\beta }^{-1}( \sigma
_{0})\,\allowbreak
_{1}h_{n_{i},\beta }^{\prime }\left( \mu _{i0},\sigma_0\right) +o_{p}(1) ,\quad i=1,2.
\end{equation*}
%
%
Hence, using (\ref{1.1}) we get
%
%
\begin{equation}
\sqrt{n_{i}}(\widehat{\mu }_{i\beta }-\mu _{i0})\underset{n_{i}\rightarrow
\infty }{\overset{\mathcal{L}}{\longrightarrow }}\mathcal{N}\left(
0,K_{11,\beta }\boldsymbol{(}\sigma_0)J_{11,\beta }^{-2}( \sigma_0) \right) ,\quad i=1,2,
\label{distmu}
\end{equation}
%
%
where
%
%
\begin{equation}
K_{11,\beta }\boldsymbol{(}\sigma_0)J_{11,\beta }^{-2}(\sigma
_{0})=\sigma_0^{2}\left( \beta +1\right) ^{3}\left( 2\beta +1\right) ^{-%
\frac{3}{2}}.  \label{KJ1}
\end{equation}
%
%
It is clear that $\widehat{\mu }_{1\beta }$ and $\widehat{\mu }_{2\beta }$ are based on two 
independent set of observations, hence, $Cov(\widehat{\mu }_{1\beta },\widehat{\mu }_{2\beta })=0$.
As $_{2}h_{n_1,n_2,\beta }^{\prime }(\widehat{\boldsymbol{\eta }}_\beta )=0$, taking a Taylor series
expansion around $\boldsymbol{\eta }_0$ we get
\begin{align}
_{2}h_{n_1,n_2,\beta }^{\prime }(\widehat{\boldsymbol{\eta }}_\beta ) =& \ 
_{2}h_{n_1,n_2,\beta }^{\prime }( \boldsymbol{\eta }_0) + \left. \frac{\partial }{\partial \mu_1}\,\allowbreak
_{2}h_{n_1,n_2,\beta }^{\prime }( \boldsymbol{\eta })\right\vert_{\boldsymbol{\eta }=\boldsymbol{\eta }_0} (\widehat{\mu }_{1\beta }-\mu _{10}) \nonumber\\
& + \left. \frac{\partial }{\partial \mu_2}\,\allowbreak _{2} h_{n_1,n_2,\beta
}^{\prime \prime }( \boldsymbol{\eta }) \right\vert_{\boldsymbol{\eta }=\boldsymbol{\eta }_0} (\widehat{%
\mu }_{2\beta }-\mu _{20}) \nonumber\\
& + \left. \frac{\partial }{\partial \sigma }\,\allowbreak _{2}h_{n_1,n_2,\beta
}^{\prime \prime }( \boldsymbol{\eta })\right\vert_{\boldsymbol{\eta }=\boldsymbol{\eta }_0} \left(
\widehat{\sigma}_\beta -\sigma_0\right) +o_{p}\left(
(n_1+n_2)^{-1/2}\right) \nonumber\\
=& \ 0.
\label{2h'}
\end{align}%
%
%
Notice that
\begin{align}
\lim_{n_1,n_2\rightarrow \infty } \left. \frac{\partial }{\partial \mu_1}\,\allowbreak _{2}h_{n_1,n_2,\beta
}^{\prime }( \boldsymbol{\eta }) \right\vert_{\boldsymbol{\eta }=\boldsymbol{\eta }_0} & = \lim_{n_1,n_2\rightarrow \infty } \frac{\partial }{%
\partial \mu_1}\left( \frac{n_1}{n_1+n_2}\,\allowbreak
_{2}h_{n_1,\beta }^{\prime }\left( \mu _{10},\sigma_0\right) +\frac{%
n_2}{n_1+n_2}\,\allowbreak _{2}h_{n_2}^{\prime }(\mu _{10},\sigma
_{0})\right)  \nonumber\\
& =\lim_{n_1,n_2\rightarrow \infty } \frac{n_1}{n_1+n_2} \lim_{n_1,n_2\rightarrow \infty } \left. \frac{\partial }{\partial \mu_1}\,\allowbreak
_{2}h_{n_1,\beta }^{\prime }\left( \mu _{1},\sigma_0\right) \right\vert_{\mu_1 = \mu_{10} }\nonumber\\
&= w \boldsymbol{J}_{12,\beta }\left( \sigma_0\right) = 0.
\label{hmu2}
\end{align}
%
%
Similarly we get
%
%
\begin{equation}
\lim_{n_1,n_2\rightarrow \infty } \left. \frac{\partial }{\partial \mu_2}\,\allowbreak _{2}h_{n_1,n_2,\beta
}^{\prime }( \boldsymbol{\eta })\right\vert_{\boldsymbol{\eta }=\boldsymbol{\eta }_0} =0.
\label{hmu1}
\end{equation}
%
%
Moreover, 
\begin{eqnarray}
\lim_{n_1,n_2\rightarrow \infty } \left. \frac{\partial }{\partial \sigma }%
\allowbreak _{2}h_{n_1,n_2,\beta }^{\prime }\left( \boldsymbol{\eta } \right)\right\vert_{\boldsymbol{\eta }=\boldsymbol{\eta }_0}  &=&\lim_{n_1,n_2\rightarrow \infty }\tfrac{%
n_1}{n_1+n_2}\,\allowbreak _{22}h_{n_1,\beta }^{\prime \prime }(\mu
_{10},\sigma_0)+\lim_{n_1,n_2\rightarrow \infty }\tfrac{n_2}{%
n_1+n_2}\,\allowbreak _{22}h_{n_2,\beta }^{\prime \prime }\left( \mu
_{20},\sigma_0\right)  \nonumber\\
&=&w\boldsymbol{J}_{22,\beta }\left( \sigma_0\right) +(1-w)\boldsymbol{J}_{22,\beta }( \sigma_0) = \boldsymbol{J}_{22,\beta } ( \sigma_0).
\label{hsigma}
\end{eqnarray}
%
Therefore, using equations (\ref{hmu2}), (\ref{hmu1}) and (\ref{hsigma}) we get from equation (\ref{2h'}) 
%
%
\begin{equation}
\sqrt{n_1+n_2}\left( \widehat{\sigma}_\beta -\sigma_0\right)
=-\boldsymbol{J}_{22,\beta }^{-1}\left( \sigma_0\right) \sqrt{n_1+n_2}%
\,\allowbreak _{2}h_{n_1,n_2,\beta }^{\prime }( \boldsymbol{\eta }_0) +o_{p}(1).
\label{sigmaL}
\end{equation}
%
%
%
Applying (\ref{1.1}) and (\ref{EQ:w}) we get
%
%
\begin{align*}
& \lim_{n_1,n_2\rightarrow \infty } \text{$E$}\left[ \sqrt{n_1+n_2}\,\allowbreak _{2}h_{n_1,n_2,\beta
}^{\prime }(\boldsymbol{\eta }_0) \right]  \\
& = \lim_{n_1,n_2\rightarrow \infty } \frac{\sqrt{n_1+n_2}}{n_1+n_2}E\left[ n_1\,\allowbreak
_{2}h_{n_1,\beta }^{\prime }\left( \mu _{10},\sigma_0\right)
+n_2\,\allowbreak _{2}h_{n_2,\beta }^{\prime }\left( \mu _{20},\sigma
_{0}\right) \right]  \\
& = \lim_{n_1,n_2\rightarrow \infty }\sqrt{\frac{n_1}{n_1+n_2}} \lim_{n_1,n_2\rightarrow \infty } E\left[ \sqrt{n_1}\,\allowbreak
_{2}h_{n_1,\beta }^{\prime }(\mu _{10},\sigma_0)\right] \\
& \ \ \ + \lim_{n_1,n_2\rightarrow \infty } \sqrt{\frac{%
n_2}{n_1+n_2}} \lim_{n_1,n_2\rightarrow \infty } E\left[ \sqrt{n_2}\,\allowbreak _{2}h_{n_2,\beta
}^{\prime }\left( \mu _{20},\sigma_0\right) \right] \\
&= 0.
\end{align*}
%
%
%
%
Similarly we also have
%
%
\begin{align*}
& \lim_{n_1,n_2\rightarrow \infty } \text{$Var$}\left[ \sqrt{n_1+n_2}\,\allowbreak
_{2}h_{n_1,n_2,\beta }^{\prime }(\boldsymbol{\eta }_0) \right]  \\
& =\lim_{n_1,n_2\rightarrow \infty } (n_1+n_2)\text{$Var$}\left[ \frac{1}{n_1+n_2}\left(
n_1\,\allowbreak _{2}h_{n_1,\beta }^{\prime }(\mu _{10},\sigma
_{0})+n_2\,\allowbreak _{2}h_{n_2,\beta }^{\prime }(\mu _{20},\sigma
_{0}\right) \right]  \\
& = \lim_{n_1,n_2\rightarrow \infty } \frac{n_1}{n_1+n_2} \lim_{n_1,n_2\rightarrow \infty } \text{$Var$}\left[ \sqrt{n_1}\,\allowbreak
_{2}h_{n_1,\beta }^{\prime }(\mu _{10},\sigma_0)\right] \\
& \ \ \ + \lim_{n_1,n_2\rightarrow \infty } \frac{n_2}{%
n_1+n_2} \lim_{n_1,n_2\rightarrow \infty } \text{$Var$}\left[ \sqrt{n_2}\,\allowbreak _{2}h_{n_2,\beta
}^{\prime }(\mu _{20},\sigma_0)\right] \\
&= w \boldsymbol{K}_{22,\beta }(\sigma_0) + (1-w) \boldsymbol{K}_{22,\beta }(\sigma_0) \\
&= \boldsymbol{K}_{22,\beta }(\sigma_0) .
\end{align*}%
%
%
%
%
Hence,
\begin{equation*}
\sqrt{n_1+n_2}\,\allowbreak _{2}h_{n_1,n_2,\beta }^{\prime }\left(
\boldsymbol{\eta }_0\right) \underset{n_1,n_2\rightarrow
\infty }{\overset{\mathcal{L}}{\longrightarrow }}\mathcal{N}\left(
0,\boldsymbol{K}_{22,\beta }(\sigma_0)\right) .
\end{equation*}%
%
%
Now, from equation (\ref{sigmaL}) we get
%
%
\begin{equation}
\sqrt{n_1+n_2}\left( \widehat{\sigma}_\beta -\sigma_0\right)
\underset{n_1,n_2\rightarrow \infty }{\overset{\mathcal{L}}{%
\longrightarrow }}\mathcal{N}\left( 0,\boldsymbol{K}_{22,\beta }(\sigma_0)\boldsymbol{J}_{22,\beta
}^{-2}(\sigma_0)\right) ,
\label{sigma1}
\end{equation}%
%
%
where
%
%
\begin{equation}
\boldsymbol{K}_{22,\beta }(\sigma_0)\boldsymbol{J}_{22,\beta }^{-2}(\sigma_0)=\sigma_0^{2}%
\frac{\left( \beta +1\right) ^{5}}{\left( \beta ^{2}+2\right) ^{2}}\left(
\frac{4\beta ^{2}+2}{(1+2\beta )^{5/2}}-\frac{\beta ^{2}}{(1+\beta )^{3}}%
\right).  \label{KJ2}
\end{equation}%
%
%
As $\boldsymbol{J}_{12,\beta }(\sigma_0)=\boldsymbol{J}_{21,\beta }(\sigma_0)=0$, it is clear that
\begin{equation*}
 \lim_{n_1,n_2\rightarrow \infty } \left. \frac{\partial^2 }{\partial \mu_1 \partial \sigma}\,\allowbreak h_{n_1,n_2,\beta
}( \boldsymbol{\eta })\right\vert_{\boldsymbol{\eta } = \boldsymbol{\eta }_0} 
= \lim_{n_1,n_2\rightarrow \infty } \left. \frac{\partial^2 }{\partial \mu_2 \partial \sigma}\,\allowbreak h_{n_1,n_2,\beta
}( \boldsymbol{\eta }) \right\vert_{\boldsymbol{\eta } = \boldsymbol{\eta }_0}
=0.
\end{equation*}
Therefore, $Cov(\widehat{\mu }_{1\beta },\widehat{\sigma}_{\beta })=Cov(\widehat{\mu }_{2\beta },\widehat{\sigma}_{\beta })=0$. 
Moreover, $Cov(\widehat{\mu }_{1\beta },\widehat{\mu }_{2\beta })=0$. Combining the results in (\ref{distmu}) and (\ref{sigma1})
we get the variance-covariance matrix of $\sqrt{\frac{n_1n_2}{n_1+n_2}}\widehat{\boldsymbol{\eta }}_{\beta }$ as follows
%
%
\begin{equation*}
\boldsymbol{\Sigma }_{w,\beta }(\sigma_0)=\left(
\begin{array}{ccc}
\left( 1-w\right) \boldsymbol{K}_{11,\beta }\boldsymbol{(}\sigma_0) \boldsymbol{J}_{11,\beta
}^{-2}\left( \sigma_0\right)  & 0 & 0 \\
0 & w \boldsymbol{K}_{11,\beta }\boldsymbol{(}\sigma_0) \boldsymbol{J}_{11,\beta }^{-2}(\sigma_0)
& 0 \\
0 & 0 & w\left( 1-w\right) \boldsymbol{K}_{22,\beta }(\sigma_0)\boldsymbol{J}_{22,\beta
}^{-2}\left( \sigma_0\right)
\end{array}%
\right) ,
\end{equation*}
where the values of the diagonal elements are given in (\ref{KJ1}) and (\ref{KJ2}). 
Hence, the theorem is proved.
\hspace*{\fill}${\blacksquare }$
\bigskip

\noindent \textbf{Proof of Theorem \ref{Th1}}: A Taylor expansion of $%
d_{\gamma }(f_{\widehat{\mu }_{1\beta },\widehat{\sigma}_\beta },f_{%
\widehat{\mu }_{2\beta },\widehat{\sigma}_\beta })$ around $\boldsymbol{\eta }_{0}$ gives%
%
%
\begin{equation*}
d_{\gamma }(f_{\widehat{\mu }_{1\beta },\widehat{\sigma}_\beta },f_{%
\widehat{\mu }_{2\beta },\widehat{\sigma}_\beta })=d_{\gamma }(f_{\mu
_{10},\sigma_0},f_{\mu _{20},\sigma_0})+\boldsymbol{t}_{\gamma
}^{T}\left( \boldsymbol{\eta }_{0}\right) (\widehat{\boldsymbol{\eta }}%
_{\beta }-\boldsymbol{\eta }_{0})+o_{p}\left( \left\Vert \widehat{%
\boldsymbol{\eta }}_{\beta }-\boldsymbol{\eta }_{0}\right\Vert \right),
\end{equation*}
%
%
where $\boldsymbol{t}_{\gamma }\left( \boldsymbol{\eta }_{0}\right) =\frac{%
\partial }{\partial \boldsymbol{\eta }}\left. d_{\gamma }(f_{\mu_1,\sigma
},f_{\mu_2,\sigma })\right\vert _{\boldsymbol{\eta }=\boldsymbol{\eta }%
_{0}}$;
%
%
the expressions of the components $t_{\gamma, i}\left( \boldsymbol{\eta }%
_{0}\right) $, $i=1,2,3$, are given in (\ref{t1})-(\ref{t3}). Hence, the result directly follows from Theorem \ref{Th0}.
\hspace*{\fill}${\blacksquare }$
\bigskip

\noindent \textbf{Proof of Theorem \ref{Th2}}: If $\mu _{10}=\mu _{20}$, it is obvious that
$d_{\gamma }(f_{\mu_{10},\sigma_0},f_{\mu _{20},\sigma_0})=0$, and $\boldsymbol{t}_{\gamma } ( \boldsymbol{\eta }_{0})=0$.
Hence, a second order Taylor expansion of $d_{\gamma }(f_{\widehat{\mu }_{1\beta },\widehat{\sigma}_\beta },f_{%
\widehat{\mu }_{2\beta },\widehat{\sigma}_\beta })$ around $\boldsymbol{\eta }_{0}$ gives
%
\begin{equation}
2d_{\gamma }(f_{\widehat{\mu }_{1\beta },\widehat{\sigma}_\beta },f_{%
\widehat{\mu }_{2\beta },\widehat{\sigma}_\beta })=(\widehat{\boldsymbol{%
\eta }}_{\beta }-\boldsymbol{\eta }_{0})^{T}\boldsymbol{A}_{\gamma }\left(
\sigma_0\right) (\widehat{\boldsymbol{\eta }}_{\beta }-\boldsymbol{\eta }%
_{0})+o_p(\left\Vert \widehat{\boldsymbol{\eta }}_{\beta }-\boldsymbol{\eta }%
_{0}\right\Vert ^{2}),
\label{gam}
\end{equation}%
%
%
where $\boldsymbol{A}_{\gamma }(\sigma_0)$ is the matrix containing the second
derivatives of $d_{\gamma }(f_{\mu_1,\sigma },f_{\mu_2,\sigma })\ $%
evaluated at $\mu_{10}=\mu_{20}$. It can be shown that
%
%
\begin{equation*}
\boldsymbol{A}_{\gamma }\left( \sigma_0\right) \boldsymbol{=}\ell
_{\gamma }(\sigma_0)\left(
\begin{array}{ccc}
1 & -1 & 0 \\
-1 & 1 & 0 \\
0 & 0 & 0%
\end{array}%
\right) ,
\end{equation*}
%
%
where
%
%
\begin{equation*}
\ell _{\gamma }(\sigma_0)=\sigma ^{-(\gamma +2)}\left( 2\pi \right) ^{-%
\frac{\gamma }{2}}\left( \gamma +1\right) ^{-\frac{1}{2}}.
\end{equation*}
%
%
Therefore, equation (\ref{gam}) simplifies to
%
%
\begin{equation*}
2d_{\gamma }(f_{\widehat{\mu }_{1\beta },\widehat{\sigma}_\beta },f_{%
\widehat{\mu }_{2\beta },\widehat{\sigma}_\beta })=\left(
\begin{pmatrix}
\widehat{\mu }_{1\beta } \\
\widehat{\mu }_{2\beta }%
\end{pmatrix}%
-%
\begin{pmatrix}
\mu _{10} \\
\mu _{20}%
\end{pmatrix}%
\right) ^{T}\boldsymbol{A}_{\gamma }^{\ast }\left( \sigma_0\right) \left(
\begin{pmatrix}
\widehat{\mu }_{1\beta } \\
\widehat{\mu }_{2\beta }%
\end{pmatrix}%
-%
\begin{pmatrix}
\mu _{10} \\
\mu _{20}%
\end{pmatrix}%
\right) +o_{p}\left( \left\Vert \widehat{\boldsymbol{\eta }}_{\beta }-%
\boldsymbol{\eta }_{0}\right\Vert ^{2}\right) ,
\end{equation*}
%
%
where
%
%
\begin{equation*}
\boldsymbol{A}_{\gamma }^{\ast }\left( \sigma_0\right) =\ell _{\gamma
}(\sigma_0)\left(
\begin{array}{cc}
1 & -1 \\
-1 & 1%
\end{array}%
\right) .
\end{equation*}
%
%
From Theorem \ref{Th0} we know that
%
%
\begin{equation*}
\sqrt{\frac{n_1n_2}{n_1+n_2}}\left(
\begin{pmatrix}
\widehat{\mu }_{1\beta } \\
\widehat{\mu }_{2\beta }%
\end{pmatrix}%
-%
\begin{pmatrix}
\mu _{10} \\
\mu _{20}%
\end{pmatrix}%
\right) ^{T}\underset{}{\overset{\mathcal{L}}{\longrightarrow }}\mathcal{N}%
\left( \boldsymbol{0}_{2},\boldsymbol{\Sigma }_{w,\beta }^{\ast }(\sigma
_{0})\right),
\end{equation*}
%
%
where
%
%
\begin{equation*}
\boldsymbol{\Sigma }_{w,\beta }^{\ast }(\sigma_0)=K_{11,\beta }%
\boldsymbol{(}\sigma_0)J_{11,\beta }^{-2}\left( \sigma_0\right) \left(
\begin{array}{cc}
1-w & 0 \\
0 & w%
\end{array}
\right).
\end{equation*}
%
%
Therefore, $\frac{2 n_1n_2}{n_1+n_2} d_{\gamma }(f_{\widehat{\mu }_{1\beta },\widehat{\sigma }_{\beta
}},f_{\widehat{\mu }_{2\beta },\widehat{\sigma}_\beta })$ has the same
asymptotic distribution (see \citealp{MR801686}) as the random variable
%
%
\begin{equation*}
\sum\limits_{i=1}^{2}\lambda _{i,\beta ,\gamma }(\sigma_0)Z_{i}^{2},
\end{equation*}
%
%
where $Z_{1}$ and $Z_{2}$ are independent standard normal variables, and
\begin{equation*}
\lambda _{1,\beta ,\gamma }(\sigma_0)=0\text{, and }\lambda _{2,\beta
,\gamma }(\sigma_0)=K_{11,\beta }\boldsymbol{(}\sigma_0)J_{11,\beta
}^{-2}\left( \sigma_0\right) \ell _{\gamma }(\sigma_0)=\lambda _{\beta
,\gamma }(\sigma_0)
\end{equation*}%
are the eigenvalues of the matrix $\boldsymbol{\Sigma }_{w,\beta }^{\ast
}(\sigma_0)\boldsymbol{A}_{\gamma }^{\ast }\left( \sigma_0\right) $.
Hence,
%
%
\begin{equation*}
\frac{2n_1n_2}{n_1+n_2}\frac{d_{\gamma
}(f_{\widehat{\mu }_{1\beta },\widehat{\sigma}_\beta },f_{\widehat{\mu }%
_{2\beta },\widehat{\sigma}_\beta })}{\lambda _{\beta ,\gamma
}\,\allowbreak (\sigma_{0} )} \underset{%
n_1,n_2\rightarrow \infty }{\overset{\mathcal{L}}{\longrightarrow }}\chi
^{2}(1).
\end{equation*}
%
%
Finally, since $\widehat{\sigma}_\beta $ is a consistent estimator of $%
\sigma $, replacing $\lambda _{\beta ,\gamma }(\sigma_0)$ by $\lambda
_{\beta ,\gamma }(\widehat{\sigma}_\beta )$ and by following 
Slutsky's theorem we obtain the desired result.
\hspace*{\fill}${\blacksquare }$

\end{document}